\documentclass[aps,pra, two column]{revtex4-2}
\usepackage{amsmath}
\usepackage[margin=1in]{geometry}
\usepackage{amssymb}
\usepackage{amsfonts}
\usepackage{amsthm}
\usepackage{thmtools,thm-restate}
\usepackage{multirow}
\usepackage{hyperref}
\usepackage{xcolor}
\usepackage{graphicx}
\newtheorem{definition}{Definition}
\newtheorem{theorem}{Theorem}
\newtheorem{corollary}{Corollary}
\newtheorem{lemma}{Lemma}
\newtheorem{remark}{Remark}
\usepackage[utf8]{inputenc}

\usepackage{phflplx}
\DeclareGraphicsExtensions{.lplx, .pdf}

\begin{document}

\title{Quantum algorithms for classical Boolean functions via adaptive measurements: Exponential reductions in space-time resources}

\author{Austin K. Daniel}
\email{austindaniel@unm.edu}
\affiliation{Department of Physics and Astronomy, Center for Quantum Information and Control, University of New Mexico, Albuquerque, New Mexico 87106, USA}
\author{Akimasa Miyake}
\email{amiyake@unm.edu}
\affiliation{Department of Physics and Astronomy, Center for Quantum Information and Control, University of New Mexico, Albuquerque, New Mexico 87106, USA}
\date{\today}

\begin{abstract}

The limited computational power of constant-depth quantum circuits can be boosted by adapting future gates according to the outcomes of mid-circuit measurements.  We formulate the computation of a variety of Boolean functions in terms of adaptive measurement-based quantum computation (MBQC) using a cluster state resource and a classical side-processor that can add bits modulo 2, so-called $l2$-MBQC. Our adaptive approach overcomes a known challenge that computing these functions in the nonadaptive setting requires a resource state that is exponentially large in the size of the computational input. Inspired by quantum signal processing constructions for computing symmetric Boolean functions via a sequence of classically conditioned quantum gates on a single qubit, we develop a unified framework to translate single-qubit quantum computation to adaptive MBQC, so that trade-offs of the space-time resources (i.e., qubit count, quantum circuit depth, classical memory size, and number of calls to the side-processor) can be optimized. In particular, we construct adaptive $l2$-MBQC algorithms that compute the mod-$p$ functions, which play a crucial role in an oracular separation of computational classes, with the best known efficient scaling in the space-time resources.
    
\end{abstract}

\maketitle

\section{Introduction}

Measurement-based quantum computation (MBQC) is a scheme by which local measurements performed on a multipartite entangled state can drive a coherent quantum computation \cite{raussendorf2001a, raussendorf2003measurement, jozsa2006introduction}.  This model can be re-interpreted as how quantum correlations boost the computational power of a sub-universal classical side-processor to full universality \cite{anders2009computational}.  Since this result, there has been a large body of literature concerning the power of MBQC with the assistance of a classical side-processor that can only compute parity functions (so-called $l2$-MBQC) \cite{raussendorf2013contextuality, hoban2011non-adaptive, frembs2018contextuality, frembs2022hierarchies, abramsky2017contextual, mori2018periodic, mackeprang2022the}.  A majority of prior work has focused on the nonadaptive case, using generalized Greenberger-Horne-Zeillinger (GHZ) states as the entangled resource.   Meanwhile, to utilize the full power of quantum computation, the measurements must be performed \emph{adaptively}, updating the basis in which each measurement is performed according measurement outcomes obtained at a previous time with the help of the classical side-processor \cite{raussendorf2001a, terhal2004adaptive}.  Recent work has leveraged adaptivity to prove novel complexity theoretic separations between classes of shallow-depth quantum and classical circuits \cite{browne2010computational, takahashi2016collapse}.

In this paper, we construct efficient adaptive $l2$-MBQC protocols that use a 1D cluster state resource \cite{briegel2001persistent} to compute a variety of Boolean functions \cite{o2014analysis}.  For clarity, we have included a sketch of the thought process leading to these improvements in Fig.~\ref{fig:flow_chart}.  In particular, our construction is inspired by past works that use classically-conditioned single-qubit circuits to compute Boolean functions \cite{ablayev2005computational, cosentino2013dequantizing, clementi2017classical, mansfield2018quantum, emeriau2022quantum}, a model we refer to as one-qubit computation (1QC).  We embed these 1QCs into cluster states by viewing them as computational tensor networks \cite{gross2007novel}. In particular, we leverage and extend the recent result of Ref.~\cite{maslov2020quantum}, using quantum signal processing \cite{low2016methodology, low2017optimal, haah2019product} to construct efficient, constant-time $l2$-MBQC protocols for computing the family of so-called mod-$p$ functions (i.e., functions that count the number of 1's in a bit string modulo a prime number $p$).  Meanwhile, any nonadaptive protocol requires an exponentially large GHZ resource state to accomplish the same task.  Our new cluster state based construction requires fewer space-time resources than previous constructions.  To clarify this point, we recast several old results for computing Boolean functions via constant-depth quantum circuits with so-called unbounded quantum fan-out gates \cite{moore1999quantum, green2002counting, hoyer2005quantum, browne2010computational, takahashi2016collapse, anand2022power} as $l2$-MBQCs implemented on cluster states and analyze their computational resource costs in terms of qubit count, classical memory size, quantum circuit depth, and number of discrete time steps in the adaptive protocol.

We see this work as timely and valuable for the following reasons.  Implementations of small circuits that are capable of demonstrating a quantum advantage have recently been used to benchmark near-term quantum devices \cite{maslov2020quantum, demirel2021correlations, sheffer2022playing, daniel2022quantum}.  Meanwhile, adaptive $l2$-MBQC requires the use of mid-circuit measurements, which have recently become a tool that is available on today's quantum computers \cite{chen2021exponential, pino2021demonstration}.  This work lies at the intersection of these ideas, showing how mid-circuit measurements can be used to demonstrate a quantum advantage for computing a particular Boolean function using fewer computational resources.  On the theoretical side, our construction is an alternative proof of an old theorem in \cite{moore1999quantum}, showing the class $\mathsf{QNC}^0[2]$ (constant-depth quantum circuits with the aid of a classical parity gate) is strictly more powerful than the class $\mathsf{AC}^0[p]$ (constant-depth classical circuits with unbounded fan-in AND, OR, and mod-$p$ gates) for any prime number $p$.  Recent break-through works have shown how one may remove the parity oracle on the quantum side for some tasks and still obtain meaningful, yet less powerful, complexity theoretic separations \cite{bravyi2018quantum, gall2018average, watts2019exponential, grier2019interactive, caha2022single}, though separating $\mathsf{QNC}^0$ and $\mathsf{AC}^0[p]$ remains elusive \cite{watts2019exponential, grier2019interactive, caha2022single}.  We see our work as a step towards separating these classes.  Finally, our work gives new insight into the study of quantum contextuality \cite{budroni2021quantum}.  Quantum contextuality is known to be the resource responsible for the quantum advantage in computing nonlinear Boolean functions via $l2$-MBQC \cite{raussendorf2013contextuality, abramsky2017contextual}.  A generic framework to describe contextual quantum computations in the language of cohomology has long been sought out \cite{raussendorf2017contextuality, raussendorf2019cohomological, aasnaess2020cohomology, frembs2018contextuality} and recent work has aimed to cast adaptive quantum computations in this framework \cite{raussendorf2022putting}.  Our work brings a plethora of new and old examples to light in the framework of $l2$-MBQC.

\begin{figure}
    \centering
    \includegraphics[width=\linewidth]{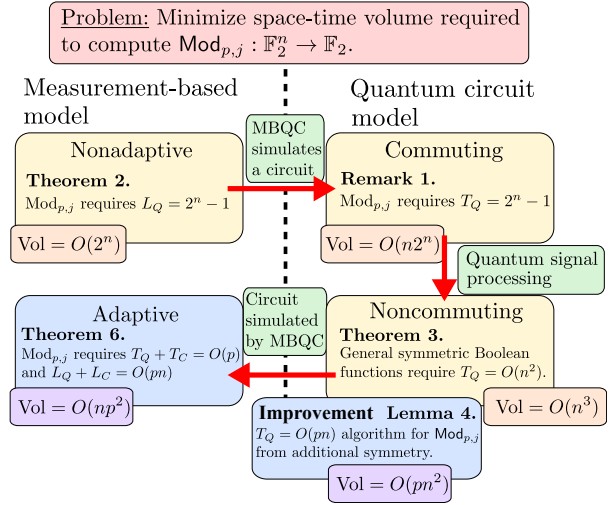}
    \caption{A flow chart of key ideas in this paper.  We wish to construct efficient algorithms to compute the Boolean function $\mathsf{Mod}_{p,j}$, defined in Eq.~(\ref{eq:mod_p_def}), on an $n$-bit input string. A unifying theme is the use of trade-offs between space-time resources when moving between the measurement-based and quantum circuit models. It was shown in Refs.~\cite{hoban2011non-adaptive, mori2018periodic} that nonadaptive $l2$-MBQC algorithms require an exponential size resource state, which we review in Thm.~\ref{thm:nmbqc_resource_cost}.  We observe in Remark~\ref{rmk:nmbqc_c1qc} that this fundamental obstruction comes from the fact that nonadaptive $l2$-MBQCs simulate the output of quantum circuits consisting of only commuting gates.  Our idea is to instead use the quantum signal processing technique developed in Ref.~\cite{maslov2020quantum}, which can compute any symmetric Boolean function via a quadratic depth single-qubit circuit, to construct efficient adaptive $l2$-MBQCs based on the 1D cluster state.  In Lemma~\ref{lem:nc1QCmod_p} we improve this construction for $\mathsf{Mod}_{p,j}$, enabling our constant-time linear-space $l2$-MBQC algorithm presented in Theorem~\ref{thm:mod_p_mbqc}.  Our algorithm has better scaling in the required space-time volume over previous algorithms (c.f., Tab.~\ref{tab:resource_cost_table}). }
    \label{fig:flow_chart}
\end{figure}

The paper is organized as follows.  In Sec.~\ref{sec:background} we define and review some fundamental facts about $l2$-MBQC, $l2$-1QC, and Boolean functions.  In Sec.~\ref{sec:prior_work} we review known results about the resource costs for computing particular Boolean functions via $l2$-MBQC in both the nonadaptive and adaptive settings.  In Sec.~\ref{sec:adaptive_l2MBQC} we give our main result.  Namely, we construct an adaptive $l2$-MBQC protocol based on the 1D cluster state that computes the mod-$p$ functions with the best known space-time resource costs.  In Sec.~\ref{sec:classical_comparison} we discuss the complexity theoretic ramifications of our construction.  Finally, in Sec.~\ref{sec:conclusion_outlook} we conclude with a brief outlook on open questions and possible extensions of our results.

\section{Background}
\label{sec:background}

In this section we first introduce the two computational models studied in this work, $l2$-MBQC and $l2$-1QC, and discuss the computational resources we will be counting for each.
A rough sketch of $l2$-MBQC and $l2$-1QC as hybrid quantum-classical circuits is given in Fig.~\ref{fig:1QC_MQC}.  We then introduce the Boolean functions we would like to compute and review some of their properties that will be of later use.

\begin{figure}
    \centering
    \includegraphics[width=\linewidth]{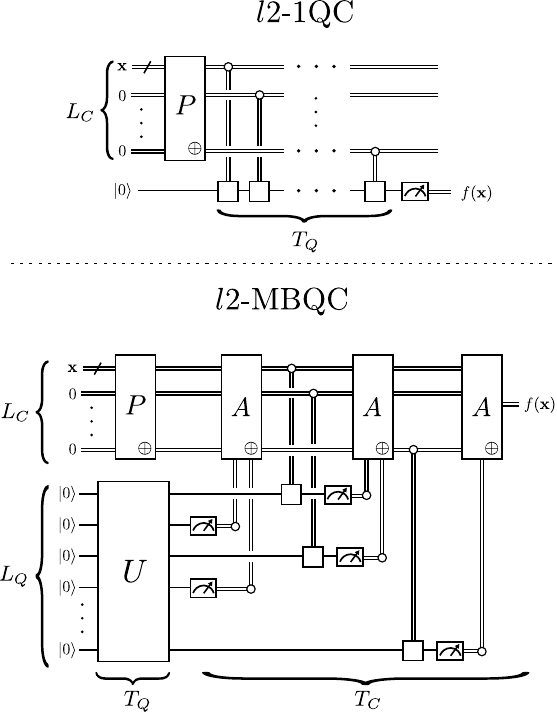}
    \caption{(above) A depiction of $l2$-1QC as a hybrid quantum-classical circuit.  Quantum and classical bits in the circuit are represented by single and double lines, respectively.  First, a single qubit initialized as $|0\rangle$ while a mod-2 linear classical side-processor (labeled by an $\oplus$) computes $L_C$ many different parity functions of the computational input $\mathbf{x}\in\mathbb{F}_2^n$ according to the binary preprocessing matrix $P$. The single qubit is evolved under a sequence of $T_Q$ many single qubit rotations, each conditioned on one of the $L_C$ many classical bits. We would like the outcome of the final measurement to equal $f(\mathbf{x})$ with probability 1. (below) A depiction of an adaptive $l2$-MBQC as a hybrid quantum-classical circuit.  Initially, an $L_Q$-qubit resource state is prepared by a depth $T_Q$ quantum circuit $U$.  Meanwhile, a mod-2 linear classical side-processor computes many different parity functions of the computational input $\mathbf{x}\in\mathbb{F}_2^n$ according to the binary preprocessing matrix $P$. A sequence of adaptive measurements are then performed on the prepared quantum state in bases determined by the preprocessing matrix $P$ along with adaptation matrix $A$, which computes parity functions of the measurement outcomes.  In the final round, the classical side-processor should output $f(\mathbf{x})$ with probability 1.}
    \label{fig:1QC_MQC}
\end{figure}

\subsection{Notation and preliminaries}

This work concerns Boolean functions \cite{o2014analysis}, i.e., maps $f:\mathbb{F}_2^n\rightarrow\mathbb{F}_2$, where $\mathbb{F}_2 := (\{0,1\}, \oplus, \cdot)$ denotes the finite field of order two under addition and multiplication modulo 2.  We will always use $\oplus$ to denote addition modulo 2 and $+$ to denote addition over the reals; however, we will sometimes add a ``mod 2" to the end of an expression for clarity.  Furthermore, $\mathbb{F}_2^n$ denotes the $n$-dimension vector space over the field $\mathbb{F}_2$ (i.e., the vector space formed by binary strings under element-wise addition modulo 2).  Finally, we denote by $[n]$ the set of natural number from 1 to $n$ (i.e., $[n]=\{1,\ldots,n\}$). 

The other major topic of this work is quantum computation.  Let $|0\rangle$ and $|1\rangle$ denote the computational basis states for a single qubit.  Let $Z$, $X$, and $Y$ denote the Pauli matrices and let $I$ be the identity matrix.  Let $H=(\sum_{j,k=0}^1 (-1)^{jk} |j\rangle\langle k|)/\sqrt{2}$ denote the Hadamard operator.  Let $\mathrm{C}Z = \sum_{j,k=0}^1 (-1)^{jk} |jk\rangle\langle jk|$. Similarly $R_\sigma(\theta) = \exp(-i\sigma\theta/2)$ is used to denote a single-qubit rotation by an angle $\theta$ about the axis $\sigma\in\{X,Y,Z\}$ of the Bloch sphere.

\subsection{$l2$-Measurement-based quantum computation}

We begin with a definition of $l2$-MBQC, a form of measurement-based quantum computing assisted by a classical computer that can only add bits modulo 2---called the mod-2 linear classical side-processor---to compute a target Boolean function $f:\mathbb{F}_2^n\rightarrow\mathbb{F}_2$.

\begin{definition}{(l2-MBQC)\cite{raussendorf2013contextuality}}
\label{def:l2_MBQC}
An l2-MBQC with classical input $\mathbf{x}\in\mathbb{F}_2^n$ and classical output $y\in\mathbb{F}_2$ is a measurement-based quantum computation driven by single-qubit measurements performed on a $L_Q$-qubit multipartite resource state $|\psi\rangle\in(\mathbb{C}^{2})^{\otimes L_Q}$ satisfying the following properties:
\begin{itemize}
    \item[1.] For each qubit $k \in [L_Q]$, there is a binary choice for the measurement basis $s_k\in\mathbb{F}_2$ corresponding to a projective measurement in the eigenbasis of an observable $O_k(s_k)$.
    \item[2.] Each measurement outcome $m_k$ is a binary number (i.e., $m_k\in\mathbb{F}_2$ $\forall k\in[L_Q]$). Namely, $O_k(s_k)^2=I$.
    \item[3.] The computational output $y\in\mathbb{F}_2$ is the parity of some subset of the measurement outcomes. I.e.,
    \begin{align}
        y = \mathbf{o}^T \mathbf{m} \oplus c \textrm{ }\mathrm{mod}\textrm{ }2,
    \end{align}
    where $\mathbf{m}\in\mathbb{F}_2^{L_Q}$ is the string of measurement outcomes, $\mathbf{o}\in\mathbb{F}_2^{L_Q}$ specifies which outcomes are relevant, and $c\in \mathbb{F}_2$ is some constant.
    \item[4.] The choice of measurement setting $\mathbf{s} \in \mathbb{F}_2^{L_Q}$ is related to the measurement outcomes $\mathbf{m}\in\mathbb{F}_2^{L_Q}$ and the input string $\mathbf{x}\in\mathbb{F}_2^n$ via
    \begin{align}
        \mathbf{s} = P\mathbf{x} \oplus A\mathbf{m} \textrm{ }\mathrm{mod}\textrm{ }2,
        \label{eq:l2MBQC_Setting}
    \end{align}
    where $P\in\mathbb{F}_2^{L_Q\times n}$ is called the preprocessing matrix and $A\in\mathbb{F}_2^{L_Q\times L_Q}$ is called the adaptation matrix. The preprocessing matrix $A$ must be lower triangular for a suitable ordering of the qubits to ensure a well-defined causal ordering of the measurements. 
\end{itemize}
We say the $l2$-MBQC computes a target Boolean function $f:\mathbb{F}_2^n\rightarrow \mathbb{F}_2$ whenever $f(\mathbf{x})=y$.
\end{definition}

We will be interested in the space-time resources required to compute a particular Boolean function via $l2$-MBQC.  The quantities of interest are sketched in Fig.~\ref{fig:1QC_MQC}.  These are: $L_Q$, the size of the quantum resource state; $L_C$, the number of unique bits output by the preprocessing matrix $P$; $T_Q$, the depth of the quantum circuit required to prepare the resource state $|\psi\rangle$; and $T_C$, the number of rounds of back and forth communication between the mod-2 linear classical side-processor and the quantum resource.  We will take as our figure of merit the total space-time volume of the algorithm, given by $(L_Q + L_C)(T_Q + T_C)$. 
The special case when $A=0$ (i.e., $A$ is the zero matrix) is called nonadaptive $l2$-MBQC.  It has been shown that the $N$-qubit GHZ state $|\mathrm{GHZ}_N\rangle = (|0\rangle^{\otimes N} + |1\rangle^{\otimes N})/\sqrt{2}$ can be used to compute any Boolean function $f:\mathbb{F}_2^n\rightarrow\mathbb{F}_2$, however, there exist worst-case Boolean functions for which $L_Q=N=2^n-1$ \cite{hoban2011non-adaptive}.  In Sec.~\ref{sec:prior_work} we will discuss how prior works have used adaptivity to overcome this exponential resource cost \cite{takahashi2016collapse, mori2018periodic}.  In Sec.~\ref{sec:adaptive_l2MBQC}, we will construct another method that is based on adaptive $l2$-MBQC with the 1D cluster state and quantum signal processing that gives even better overall scaling in the space-time resource costs. 

\subsection{One-qubit quantum computations}

We now define $l2$-1QC, which is a type of limited-space quantum computation. Here, a single qubit of active computational space is evolved under a sequence of $T_Q$ many unitaries each conditioned on one of $L_C$ many bits of classical memory in order to compute a target Boolean function.

\begin{definition}{($l2$-1QC)}
An $l2$-1QC with classical input $\mathbf{x}\in\mathbb{F}_2^n$ and classical output $y\in\mathbb{F}_2$ is a single-qubit quantum circuit supplemented by a $L_C$-bit ``read-only" classical memory register. The circuit contains $T_Q$ many gates applied sequentially that satisfy the following properties:
\begin{itemize}
    \item[1.] For each gate $k\in[T_Q]$, there is a corresponding binary choice $s_k\in\mathbb{F}_2$ for the unitary to apply, denoted $U_k(s_k)\in\mathrm{U}(2)$.
    \item[2.] The collection of classical bits $\mathbf{s}\in\mathbb{F}_2^{T_Q}$ is related to the computational input $\mathbf{x}\in\mathbb{F}_2^n$ via
    \begin{align}
        \mathbf{s} = P\mathbf{x} \textrm{ }\mathrm{mod}\textrm{ }2,
    \end{align}
    where $P\in\mathbb{F}_2^{T_Q\times n}$ is called the preprocessing matrix.
    \item[3.] The computational output $y\in\mathbb{F}_2$ is the measurement outcome obtained for the single qubit at the end of the circuit.
\end{itemize}
We say the $l2$-1QC computes a target Boolean function $f:\mathbb{F}_2^n\rightarrow\mathbb{F}_2$  $f(\mathbf{x})=y$.  In the case that each unitary is simply conditioned on the individual bits from the string $\mathbf{x}$, we will drop the prefix ``$l2$" and simply call it a 1QC since no preprocessing of the input is required. 
\end{definition}

From the above definition, a $l2$-1QC that computes a Boolean function $f:\mathbb{F}_2^n\rightarrow\mathbb{F}_2$ necessarily satisfies
\begin{align}
    \langle y|U_{T_Q}(s_{T_Q}) \cdots U_1(s_1) |0\rangle \propto \delta_{f(\mathbf{x})}^y.
    \label{eq:1QC_Def}
\end{align}
The special case when all the unitaries $U_j(s_j)$ commute is called a commuting $l2$-1QC.  Otherwise, we call it a noncommuting $l2$-1QC.  We will be interested in the space-time resource costs $L_C$ and $T_Q$, sketched in Fig.~\ref{fig:1QC_MQC}, required to compute a particular Boolean function.

\subsection{Boolean functions}
In this section we will give a brief overview of some properties of Boolean functions we will use in this paper.  A Boolean function $f$ is a map $f: \mathbb{ F }_2^n  \rightarrow  \mathbb{ F }_2$.  Any Boolean function may be expressed in terms of its truth-table (i.e., a list of $2^n$ binary numbers $f(\mathbf{x})$ corresponding the value $f$ evaluates to for each string $\mathbf{x}\in\mathbb{F}_2^n$) or equivalently in terms of its algebraic normal form (ANF) (i.e., a polynomial $f(\mathbf{x})\in \mathbb{F}_2[x_1,\ldots,x_{n}]$).    A Boolean function $f$ can be expressed in ANF as
\begin{align}
    f(\mathbf{x}) = \bigoplus_{S\subseteq[n]} a_S \prod_{j\in S}x_j,
\end{align}
where $a_S \in \mathbb{ F }_2$ for each $S\subseteq[n]$.  The degree of $f$, denoted $\mathrm{deg}(f)$, is the maximal degree of any monomial in the ANF. 

In this paper, we will primarily be interested in symmetric Boolean functions. A symmetric Boolean function is a Boolean function that is invariant to permutations of the bits in the input string $\mathbf{x}\in\mathbb{F}_2^n$.  Consequentially, these functions only depend on the hamming weight of the input string $|\mathbf{x}|=x_1+\cdots+x_n$ and their algebraic normal forms satisfy $a_S = a_{S'}$ whenever $|S|=|S'|$.

For different computational models there are different ways to characterize the complexity of computing a particular Boolean function.  Let us briefly look at each and comment on their physical meaning.  

The weakest classical computational model we consider is the mod-2 linear classical-side processor itself.  While it cannot compute nonlinear Boolean functions deterministically it can do so probabilistically with a success probability depending on the fraction of inputs the target function $f(\mathbf{x})$ agrees with a linear function $\mathbf{k}\cdot\mathbf{x}$, where $\mathbf{k}\in\mathbb{F}_2^n$ and $\mathbf{k}\cdot\mathbf{x} = k_1x_1 \oplus \cdots \oplus k_n x_n$ denotes the dot product of two binary vectors.  This is quantified by the Boolean function's largest Fourier coefficient, $\hat{f}_\mathrm{max} = \max_{\mathbf{k}\in \mathbb{F}_2^n}\{|\hat{f}(\mathbf{k})|\}$ where
\begin{align}
    \hat{f}(\mathbf{k}) = \frac{1}{2^n} \sum_{\mathbf{x}\in\mathbb{F}_2^n}(-1)^{f(\mathbf{x}) + \mathbf{k}\cdot\mathbf{x}}.
    \label{eq:walsh_hadamard_def}
\end{align}
The above formula is often called the Walsh-Hadamard transform of the Boolean function $f$.  It constitutes a Boolean analogue of the Fourier transform.  The probability for this model to produce the correct computational output given an input string $\mathbf{x}\in\mathbb{F}_2^n$ drawn uniformly at random is then
\begin{align}
    \mathrm{pr}(y=f(\mathbf{x})) \leq \frac{1 + \hat{f}_\mathrm{max}}{2}.
\end{align}

The first quantum computational model we will consider is that of nonadaptive $l2$-MBQC.  While this model can compute any Boolean function determistically, it has been shown in Ref.~\cite{mori2018periodic} that the number of qubits required depends on the function's periodic Fourier sparsity.  Namely, any Boolean function can be expressed via a periodic Fourier decomposition; i.e., via a sum of the $2^n$ linear functions $\mathbf{s}\cdot\mathbf{x}$ weighted by corresponding real angles $\phi_{\mathbf{s}}\in[0,2\pi)$ such that
\begin{align}
    (-1)^{f(\mathbf{x}) + f(\mathbf{0})} = \cos\left(\sum_{\mathbf{s}\in\mathbb{F}_2^n} (\mathbf{s}\cdot\mathbf{x})\phi_\mathbf{s}\right).
    \label{eq:periodic_fourier}
\end{align}
Ref.~\cite{mori2018periodic} showed that while such an expansion always exists, it is not unique.  In Appendix~\ref{sec:boolean_appendix} we show how to obtain all such periodic Fourier decompositions by solving a system of linear equations with different inhomogenous terms.  Denote the set of all such functions $\phi^{(f)}:\mathbb{F}_2^n\rightarrow [0,2\pi]$ that specify a periodic Fourier decomposition of $f$ as $\Phi_f$.  The Periodic-Fourier sparsity (denoted $\hat{p}_f$) is defined as the minimum number of nonzero values that any $\phi_\mathbf{s}\in\Phi_f$ can have. Namely,
\begin{align}
    \hat{p}_f = \min_{\phi_f\in\Phi_f} |\mathrm{supp}\{\phi_f(\mathbf{s})\}|.
\end{align}
In the next section we will show that nonadaptive $l2$-MBQC of a Boolean function $f$ requires $L_Q=L_C=\hat{p}_f$.  Furthermore, noncommutative $l2$-1QC of the same function requires $T_Q=L_C=\hat{p}_f$.

A Boolean function that naturally appears in the context of $l2$-MBQC \cite{hoban2011non-adaptive} and $l2$-1QC \cite{clementi2017classical} is the pairwise-AND function, denoted $C_n^2:\mathbb{F}_2^n\rightarrow\mathbb{F}_2$.  This function is defined as
\begin{align}
    C_n^2(\mathbf{x}) = 
    \begin{cases}
    0 &\textrm{if }|\mathbf{x}| = 0\textrm{ or 1 mod 4} \\
    1 &\textrm{if }|\mathbf{x}| = 2\textrm{ or 3 mod 4}
    \end{cases}
\end{align}
and has ANF 
\begin{align}
    C_n^2(\mathbf{x}) = \bigoplus_{j=1}^{n-1} \bigoplus_{k=j+1}^n x_j x_k.
\end{align}
This function is also known as the second-least-significant-bit (or SLSB) function since it computes the second least significant bit in the binary representation of the integer number $|\mathbf{x}|$.  
This function has largest Fourier amplitude $|\hat{f}|_\mathrm{max} = 2^{-n/2}$.  Furthermore, its periodic Fourier sparsity is $\hat{p}_{f} = n+1$ and the corresponding periodic Fourier decomposition is given by
\begin{align}
C_n^2(\mathbf{x}) = \cos\left( \frac{\pi}{2} \sum_{\substack{S\subseteq [n] \\ |S|= 1,n}}  \left(\bigoplus_{j\in S} x_j \right) \right).
\end{align}

A family of Boolean functions that are important in classical circuit complexity theory are the mod-$p$ functions \cite{smolensky1987algebraic}, denoted $\mathsf{Mod}_{p,j}:\mathbb{F}_2^n\rightarrow\mathbb{F}_2$ for each $j\in\{0,\ldots,p-1\}$. The mod-$p$ functions effectively count the number of 1's in a string $\mathbf{x}\in\mathbb{F}_2^n$ modulo $p$.  Namely,
\begin{align}
    \mathsf{Mod}_{p,j}(\mathbf{x}) =
    \begin{cases}
    0 &\textrm{ if }|\mathbf{x}|=j\textrm{ mod }p \\
    1 &\textrm{ otherwise}
    \end{cases}.
    \label{eq:mod_p_def}
\end{align}
The ANF for the mod-$3$ function is
\begin{align}
    \mathsf{Mod}_{p,j}(\mathbf{x}) =   \bigoplus_{\substack{S\subseteq[n] \\ |S|\neq 3-j \textrm{ mod } 3}} \prod_{j\in S} x_j.
    \label{eq:mod_3_anf}
\end{align}
We prove this fact in Appendix~\ref{sec:mod3ANF} and furthermore prove the following fact.
\begin{lemma}
\label{lem:mod_p_anf}
For each possible integer $n$ and prime number $p$ at least one of the functions $\mathsf{Mod}_{p,j}:\mathbb{F}_2^n\rightarrow\mathbb{F}_2$ has degree-$n$ ANF.
\end{lemma}
\noindent This fact has been mentioned before in the literature, though we provide a direct proof in the appendix for clarity.  As we will see in the next subsection, the above lemma implies that for every $n$ the periodic Fourier sparsity is $\hat{p}_f = 2^n -1$ for at least one of the mod-$p$ functions.

One function that is known to have worst-case resource scaling in nonadaptive $l2$-MBQC is the $n$-tuple AND function, which we denote $\mathsf{AND}_n:\mathbb{F}_2^n\rightarrow\mathbb{F}_2$. It has truth table
\begin{align}
    \mathsf{AND}_n(\mathbf{x}) =
    \begin{cases}
    1 &\textrm{if }\mathbf{x} = \mathbf{1} \\
    0&\textrm{otherwise}
    \end{cases}
\end{align}
and ANF
\begin{align}
    \mathsf{AND}_n(\mathbf{x}) = \prod_{j=1}^n x_j.
\end{align}
This function has largest Fourier amplitude $|\hat{f}|_\mathrm{max} = 1 - 2^{1-n}$, which can readily be seen by the fact that the constant function $f(\mathbf{x})=0$ agrees with $\mathsf{AND}_n$ on all but one input.  Furthermore, its periodic Fourier sparsity is $\hat{p}_{f} = 2^n-1$ and the corresponding Periodic-Fourier decomposition for this function \cite{mori2018periodic} is given by
\begin{align}
\mathsf{AND}_n(\mathbf{x}) = \cos\left( \frac{\pi}{2^{n-1}} \sum_{\substack{S\subseteq [n] \\ S\neq \varnothing}} (-1)^{|S|-1} \left(\bigoplus_{j\in S} x_j \right) \right).
\end{align}
A related function that also has worst-case resource scaling is the $n$-tuple OR function, denoted $\mathsf{OR}_n:\mathbb{F}_2^n \rightarrow \mathbb{F}_2$.  This function has truth table
\begin{align}
\mathsf{OR}_n(\mathbf{x}) =
\begin{cases}
0 &\textrm{if } \mathbf{x} = \mathbf{0} \\
1 &\textrm{otherwise} 
\end{cases}.
\end{align}
The two functions are related via de-Morgan's law, $\mathsf{OR}_n(\mathbf{x}) = \mathsf{AND}_n(\mathbf{x}\oplus \mathbf{1}) \oplus 1$, and thus any $l2$-MBQC algorithm to compute $\mathsf{OR}_n$ gives an algorithm to compute $\mathsf{AND}_n$. The Periodic-Fourier sparsity of this function is $\hat{p}_f = 2^n - 1$ and the corresponding Periodic-Fourier decomposition \cite{mori2018periodic} is
\begin{align}
\label{eq:or_n_pfd}
\mathsf{OR}_n(\mathbf{x}) = \cos\left( \pi \sum_{\substack{S\subseteq [n] \\ S\neq\varnothing}}  \phi_S\left( \bigoplus_{j\in S} x_j \right) \right)
\end{align}
where $\phi_S = (-1)^{|S|-1} (2^{n-|S|+1}-1)/2^{n-1}$.

\section{Exposition of prior results}
\label{sec:prior_work}

In this section, we review several known quantum algorithms for computing various Boolean functions via $l2$-MBQC and $l2$-1QC.  We first review in Sec.~\ref{sec:NMBQC} the nonadaptive $l2$-MBQC algorithm introduced in Ref.~\cite{hoban2011non-adaptive} and show it requires an exponential size resource state (i.e., $L_Q=2^n-1$) for a variety of functions including $\mathsf{AND}_n$, $\mathsf{OR}_n$, and $\mathsf{Mod}_{p,j}$.  We then review in Sec.~\ref{sec:OR_Red_MBQC} a variety of adaptive $l2$-MBQC algorithms introduced in Ref.~\cite{takahashi2016collapse} based on the so-called ``OR-reduction" that are capable of computing symmetric Boolean functions efficiently.  Finally, we draw a correspondence between nonadaptive $l2$-MBQC and commutative $l2$-1QCs in Sec.~\ref{sec:1QC_review} and show how a recent construction \cite{maslov2020quantum} of noncommutative $l2$-1QCs based on quantum signal processing overcomes this, which we will build upon more in Sec.~\ref{sec:adaptive_l2MBQC}.

\subsection{Nonadaptive $l2$-MBQC and periodic Fourier decompositions}
\label{sec:NMBQC}

Ref.~\cite{hoban2011non-adaptive} showed that any Boolean function $f:\mathbb{F}_2^n\rightarrow\mathbb{F}_2$ can be computed via nonadaptive $l2$-MBQC on a sufficiently large GHZ state, where the $N$-qubit GHZ state is defined as $|\textrm{GHZ}_N\rangle = (|0\rangle^{\otimes N} + |1\rangle^{\otimes N})/\sqrt{2}$.  The content of this result is summarized in the following theorem.
\begin{theorem}{\cite{hoban2011non-adaptive}}
\label{thm:NMBQC_universal}
Any function $f:\mathbb{F}_2^n\rightarrow\mathbb{F}_2$ with periodic Fourier decomposition $\cos(\sum_{\mathbf{p}\in\mathbb{F}_2^n} (\mathbf{p}\cdot\mathbf{x})\phi_j)$ can be computed via nonadaptive $l2$-MBQC on the $\hat{p}_f$-qubit GHZ state.
\end{theorem}
\begin{proof}
Let $\{\mathbf{p}_j\}_{j=1}^{\hat{p}_f}$ index the linear functions appearing in the periodic Fourier decomposition of $f$ with periodic Fourier sparsity $\hat{p}_f$. Let $X(\boldsymbol{\theta}) = \bigotimes_{j=1}^N X_j(\theta_j)$ where $X(\theta) = \cos(\theta) X + \sin(\theta) Y$ and $\boldsymbol{\theta} \in [0,2\pi)^N$.  One can readily verify that $\langle \mathrm{GHZ}_N | X(\boldsymbol{\theta}) | \mathrm{GHZ}_N\rangle = \cos(\sum_{j=1}^N\theta_j)$. It then follows that if $N = \hat{p}_f$ and the angles $\{\theta_j = (\mathbf{p}_j\cdot\mathbf{x})\phi_j\}_{j=1}^{\hat{p}_f}$ form a periodic Fourier decomposition as in Eq.~(\ref{eq:periodic_fourier}), then the above measurements performed on the state $|\mathrm{GHZ}_{\hat{p}_f}\rangle$ constitute a $l2$-MBQC that computes the function $f$ with $P_{j,k}=(\mathbf{p}_j)_k$, $A=0$, $\mathbf{o}=\mathbf{1}$, and $c=f(\mathbf{0})$.
\end{proof}
\noindent It follows from the above proof that this scheme has resource costs $L_Q = \hat{p}_f$, $L_C = \hat{p}_f$, $T_Q = T_{\mathrm{GHZ}}$, and $T_C = 1$.  Here $T_{\mathrm{GHZ}_{L_Q}}$ is the quantum circuit depth required to prepare the GHZ state on $L_Q$ many qubits.  With the assistance of a mod-2 linear classical side-processor this can be done in constant time \cite{briegel2001persistent}.

Functions with periodic Fourier sparsity $\hat{p}_f=\mathrm{poly}(n)$ can be computed efficiently via nonadaptive $l2$-MBQC \cite{hoban2011non-adaptive}.  An example of such a function is the pairwise-AND function which has $\hat{p}_f = n+1$.  This $l2$-MBQC protocol is equivalent to the so-called Mermin inequalities \cite{mermin1990extreme}.

In the worst-case, the periodic Fourier sparsity of a function can be $\hat{p}_f = 2^n-1$.  The corresponding nonadaptive $l2$-MBQC protocol thus requires exponentially many qubits prepared in a GHZ state.  The following theorem characterizes a family of Boolean functions with this worst-case resource scaling.
\begin{restatable}{lemma}{pfslem}{\cite{mori2018periodic}}
\label{lem:pfslem_label}
Any function $f:\mathbb{F}_2^n\rightarrow\mathbb{F}_2$ with maximal degree ANF (i.e. any function containing the monomial $x_1\cdots x_n$ in its ANF) has periodic Fourier sparsity $\hat{p}_f=2^n-1$.
\end{restatable}
We give a new proof of Lemma~\ref{lem:pfslem_label} in Appendix~\ref{sec:proof_of_ntic_lemma} by explicitly determining the structure of all possible periodic Fourier decompositions of a given function.  From it, we immediately obtain the following fact.
\begin{theorem}{\cite{mori2018periodic}}
\label{thm:nmbqc_resource_cost}
Any nonadaptive $l2$-MBQC of a Boolean function $f:\mathbb{F}_2^n\rightarrow\mathbb{F}_2$ with degree $n$ ANF requires a $(2^n-1)$-qubit GHZ state.
\end{theorem}
\noindent It then follows from Lemma~\ref{lem:mod_p_anf} that the mod-$p$ functions constitute a family of functions that require $L_Q= 2^n-1$ to compute via nonadaptive $l2$-MBQC.  As we will see in the next section, this difficulty can be circumvented in the adaptive regime.

\subsection{Adaptive $l2$-MBQC and OR-reductions}
\label{sec:OR_Red_MBQC}

There is a large body of literature concerning constant-depth quantum circuits with so-called unbounded fan-out gates (See Ref.~\cite{anand2022power} for a nice overview).  The quantum fan-out gate is defined as the $n$-qubit unitary $F$ satisfying $F|x_1,x_2,\ldots,x_n\rangle = |x_1,x_1\oplus x_2,\ldots, x_1\oplus x_n\rangle$ in the computational basis.  It was observed in Ref.~\cite{browne2010computational} that the gate $F$ can be implemented in a measurement-based manner by a constant-depth circuit with one round of measurement, mod-2 linear classical side-processing, and feedback (i.e., $T_Q + T_C = O(1)$).  This scheme follows from the standard measurement-based reduction of a 1D cluster state to a GHZ state \cite{briegel2001persistent}, which is obtained by measuring every other qubit in the Pauli-$X$ basis.

Later, Ref.~\cite{takahashi2016collapse} showed that constant-depth quantum circuits with fan-out can compute any symmetric Boolean function.  The underlying mechanism of their proof is the so-called OR-reduction introduced in Ref.~\cite{hoyer2005quantum}, which uses a constant-depth quantum circuit with fan-out to reduce the problem of computing $\mathsf{OR}_n$ to that of computing $\mathsf{OR}_{\lceil\log_2(n)\rceil}$.  The latter can be efficiently computed using Thm.~\ref{thm:NMBQC_universal} and Eq.~\ref{eq:or_n_pfd}.  Ref.~\cite{mori2018periodic} rephrased this as a two round adaptive $l2$-MBQC using $O(\log(n))$ many $O(n)$-qubit GHZ states.  For completeness, we review this scheme in detail and rephrase it as an adaptive $l2$-MBQC performed on a 1D cluster state in Appendix~\ref{sec:previous_work_as_cluster_mbqc}.

While the work of Refs.~\cite{takahashi2016collapse, mori2018periodic} give efficient adaptive $l2$-MBQC algorithms that compute $\mathsf{Mod}_{p,j}$ with $L_C + L_Q = O(n^2\log(n))$ and $T_Q + T_C = O(1)$, their construction is not optimal in either qubit count or overall space-time volume of the algorithm.  An older result by Moore~\cite{moore1999quantum} constructed constant-depth quantum circuits with fan-out that reduce the computation of $\mathsf{Mod}_{p,0}$ to the problem of computing $\mathsf{OR}_{\lceil \log(p) \rceil}$.  Unlike the previous reduction, implementation of this circuit via adaptive $l2$-MBQC requires a 2D cluster state on a $\lceil \log_2(p) \rceil \times O(n)$ grid due to subtleties of implementing discrete Fourier transforms modulo a prime with qubits.  We give explicit details for this implementation in Appendix~\ref{sec:previous_work_as_cluster_mbqc} and show the overall resource costs $L_C+L_Q = O(n\log(p))$ and $T_C+T_Q = O(p^2\log^3(p))$.  On the other hand, in Sec.~\ref{sec:adaptive_n_n} we will give an alternative $l2$-MBQC algorithm using 1D cluster states with improved scaling in the overall space-time volume.  This construction is inspired by recent progress in engineering 1QC that compute a target symmetric Boolean function, which we now review.

\subsection{Efficient 1QCs for symmetric Boolean functions}
\label{sec:1QC_review}

We first mention that all the results on nonadaptive $l2$-MBQC with a GHZ resource state in Theorem~\ref{thm:NMBQC_universal} and Theorem~\ref{thm:nmbqc_resource_cost} also apply to commutative $l2$-1QCs.  The two models are closely related, which can readily been seen from old ideas in quantum circuit parallelization~\cite{moore1999quantum} or by viewing the state $|\mathrm{GHZ}_N\rangle$ as a computational tensor network \cite{gross2007novel}. In particular, let $|m(\theta)\rangle$ denote the eigenstate of $X(\theta)$ with eigenvalue $(-1)^{m}$ and let $|\mathbf{m}(\boldsymbol{\theta})\rangle = \bigotimes_{j=1}^N |m_j(\theta_j)\rangle$. 
\begin{remark}
\label{rmk:nmbqc_c1qc}
A single round of single-qubit nonadaptive measurements in the eigenbases of observables $\{X_j(\theta_j)\}_{j=1}^N$ on the state $|\mathrm{GHZ}_N\rangle$ simulates a single-qubit quantum circuit consisting of commuting rotations $\{R_X(\theta_j)\}_{j=1}^N$.  In particular,
\begin{align}
\label{eq:GHZ_1QC_Correspond}
     |\langle \mathbf{m}(\boldsymbol{\theta})|\mathrm{GHZ}_N\rangle|^2 &\propto |\langle \bigoplus_{j=1}^N m_j | \prod_{j=1}^N R_X( \theta_j) |0\rangle |^2.
\end{align}
\end{remark}
The above remark follows from old ideas regarding the parallelization of commuting gates \cite{moore1998some}; however, we show this in Appendix~\ref{sec:MBQC-1QC_correspondence} using a tensor network representation of the GHZ state.
Eq.~\ref{eq:GHZ_1QC_Correspond} says that if the angles $\mathbf{\theta}$ are chosen so that we get a nonadaptive $l2$-MBQC that computes a Boolean function $f$, then the angles can also be used to construct a commutative $l2$-1QC that computes the same function.  It follows that commutative $l2$-1QC can be used to compute any Boolean function $f:\mathbb{F}_2^n\rightarrow\mathbb{F}_2$ with resource costs $L_C = T_Q = \hat{p}_f = \Omega(2^n)$.

Meanwhile, noncommuting 1QCs (i.e., without the classical side-processor), or rather classically conditioned single-qubit quantum circuits that compute a particular Boolean function, have a long history of study.  Abaylev~\cite{ablayev2005computational} and later Cosentino~\cite{cosentino2013dequantizing} showed that one can use either the $A_5$ subgroup of $\mathrm{SU}(2)$ or the subgroup generated by $\{H,T\}$
in combination with Barrington's theorem \cite{barrington1988finite, barrington1989bounded} to compute any symmetric Boolean function with resource costs $L_C=n$ and $T_Q=\mathrm{poly}(n)$. However, these methods work by dissecting a classical circuit of $\mathsf{NOT}$ and 2-bit $\mathsf{AND}$ gates and effectively achieve the same resource scaling as a related classical model, so-called width-5 classical branching programs.

More recently, Maslov \textit{et al.}~\cite{maslov2020quantum} showed how one can make better use of the single-qubit state space with an approach based on quantum signal processing \cite{haah2019product}. In particular, they show the following.
\begin{theorem}[\cite{maslov2020quantum}]
\label{thm:sym_boole_qsp_1qc}
Any symmetric Boolean function $f:\mathbb{F}_2^n\rightarrow\mathbb{F}_2$ can be computed by a 1QC with resource costs $L_C=n$ and $T_Q = 4n^2 + 5n + 2$.
\end{theorem}
\begin{proof}
We simply sketch the ideas here and elaborate on the details in Appendix~\ref{sec:Previous_1QC_Work}. Using the so-called quantum signal processing decomposition, it is shown that given a target symmetric Boolean function $f:\mathbb{F}_2^n\rightarrow\mathbb{F}_2$ there exists a collection of $4n+1$ angles $\{\xi_j\in[0,2\pi)\}_{j=1}^{4n+1}$ such that
\begin{align}
    U(\mathbf{x}) = \prod_{j=1}^{4n+1} \left( R_Z(\xi_{j}) R_X\left(\frac{|\mathbf{x}| \pi}{n+1}\right) R_Z(\xi_{j})^\dagger \right). \label{eq:QSP_sym_Boole}
\end{align}
satisfies $U(\mathbf{x})=(iX)^{f(\mathbf{x})}$.  In particular, this implies that $|\langle y | U(\mathbf{x}) |0\rangle|^2 = \delta_{f(\mathbf{x})}^{y}$.  Thus $U(\mathbf{x})$ constitutes a 1QC of the function $f$.  Furthermore, notice that each rotation $R_X\left(\frac{|\mathbf{x}| \pi}{n+1}\right)$ is actually a sequence of $n$ many commuting rotations conditioned on each bit from the string $\mathbf{x}\in\mathbb{F}_2^n$.  Moreover, by combining sequential unconditioned $Z$-rotations, $U(\mathbf{x})$ can be written as a sequence of $n(4n+1) + 4n + 2$ many rotations, giving the aforementioned resource costs.
\end{proof}
\noindent It was mentioned in Ref.~\cite{maslov2020quantum} that although this construction works it is not always optimal since $f$ may have additional symmetries.  In Sec.~\ref{sec:adaptive_n_n} we will use this fact to construct 1QCs that compute $\mathsf{Mod}_{p,j}$ with resource cost $T_Q = (2p-1)n + 2p$.

\section{Adaptive $l2$-MBQC of mod-$p$ functions}
\label{sec:adaptive_l2MBQC}

In this section we construct adaptive $l2$-MBQC protocols that compute the mod-$p$ functions in constant-time using a linear size 1D cluster state.  We first review basic properties of adaptive $l2$-MBQC with the 1D cluster state in Sec.~\ref{sec:1D_cluster_MBQC}.  We then show in Theorem~\ref{thm:l2-MBQC_QSP_1DC} how a naive implementation of the 1QCs from Theorem~\ref{thm:sym_boole_qsp_1qc} gives efficient $l2$-MBQC that compute $\mathsf{Mod}_{p,j}$ in linear time.  We then improve on this construction.   In Lemma~\ref{lem:nc1QCmod_3} and Theorem~\ref{thm:aMQC_mod3} we give a simplified protocol to compute $\mathsf{Mod}_{3,0}:\mathbb{F}_2^n\rightarrow\mathbb{F}_2$.  This adaptive $l2$-MBQC scheme is depicted as a hybrid quantum-classical circuit in Fig.~\ref{fig:mod_3_circuit}.  In Lemma~\ref{lem:nc1QCmod_p} and Theorem~\ref{thm:mod_p_mbqc} we give a general construction based on quantum signal processing that computes $\mathsf{Mod}_{p,j}:\mathbb{F}_2^n\rightarrow\mathbb{F}_2$ for any odd number $p$ and $j=0,\ldots,p-1$.  For clarity, the resource costs for our construction are compared in Tab.~\ref{tab:resource_cost_table} with those of the previous constructions reviewed in Sec.~\ref{sec:prior_work}.

\begin{table*}[]
    \centering
    
    \begin{tabular}{|c|c|c|c|c|c|}
    \hline
	\multirow{2}{4em}{Algorithm} & \multicolumn{4}{c|}{$l2$-MBQC resource cost} & \multirow{2}{4em}{Complexity} \\
	\cline{2-5}
        &  $L_Q$ & $T_Q$ & $L_C$ & $T_C$ &  \\
      \hline
      \hline
       Periodic Fourier \cite{hoban2011non-adaptive}  & $2^n-1$ & $3$ & $2^n-1$ & $2$  & $\mathsf{NCHVM}\subset \mathsf{QNC}^0[2]$ \\
       Barrington's Thm. \cite{ablayev2005computational, cosentino2013dequantizing}& $O(\mathrm{poly}(n))$ & $3$& $O(\mathrm{poly}(n))$& $O(\mathrm{poly}(n))$ & $\mathsf{NCHVM}\subset \mathsf{QNC}^0[2]$  \\
       OR-reduction \cite{takahashi2016collapse, mori2018periodic} & $\Theta(n^2\log(n))$ & $3$ & $\Theta(n\log(n))$ & $3$ & $\mathsf{ACC}^0\subseteq \mathsf{QNC}^0[2]$  \\
       Quantum Signal Processing \cite{maslov2020quantum} & $\Theta(n^2)$ & $3$ & $\Theta(n)$ & $\Theta(n)$ & $\mathsf{NCHVM}\subset \mathsf{QNC}^0[2]$  \\
       Moore's counting circuit \cite{moore1999quantum} & $O(n\log(p) + p^2\log^3(p))$ & $5$ & $\Theta(n\log(p))$ & $O(p^2\log^3(p))$ & $\mathsf{QNC}^0[2]\not\subset \mathsf{AC}^0[q]$  \\
       This work & $\Theta(pn)$ & $3$ & $n+2$ & $\Theta(p)$ & $\mathsf{QNC}^0[2] \not\subset \mathsf{AC}^0[q]$ \\
       \hline
    \end{tabular}
    
    \caption{Summary of the resource costs required for previous algorithms that compute the mod-$p$ functions.  In the right-most column labeled ``Complexity" we list relevant complexity theoretic separations that each quantum algorithm is capable of demonstrating as discussed in Sec.~\ref{sec:classical_comparison}.  $\mathsf{NCHVM}$ denotes the family of noncontextual hidden variable models, which corresponds to the power of the mod-2 linear classical side-processor by itself.  In the separation $\mathsf{QNC}^0[2]\not\subset\mathsf{AC}^0[q]$, $q$ can be taken to be any prime number less than $p$.}
    \label{tab:resource_cost_table}
\end{table*}

\subsection{Adaptive $l2$-MBQC with the 1D cluster state}
\label{sec:1D_cluster_MBQC}

Before we get into the main result, we first review a few facts about the 1D cluster state and show how it can simulate arbitrary single-qubit quantum circuits via adaptive $l2$-MBQC.

The $N$-qubit 1D cluster state is defined by the following quantum circuit
\begin{align}
    |\mathrm{1DC}_N\rangle = \left(\prod_{j=1}^{N-1} \mathrm{C}Z_{j,j+1} \right) H^{\otimes N} |0\rangle^{\otimes N}.
\end{align}
This state is prepared by a depth-3 circuit of Hadamard and nearest-neighbor $\mathrm{C}Z$-gates on a 1D chain.  
This state is sometimes called a ``one-qubit universal" resource for MBQC since the output of any one-qubit quantum circuit can computed via a sequence of adaptive single-qubit measurements.  Namely, consider a $(2N+1)$-qubit 1D cluster state and suppose each qubit $j\in[2N+1]$ is measured in the eigenbasis of $X(\theta_j)$ for some $\theta_j\in[0,2\pi)$. The probability to observe the outcome $\mathbf{m}\in\mathbb{F}_2^{2N+1}$ satisfies 
\begin{align}
    &|\langle \mathbf{m}(\boldsymbol{\theta})| \mathrm{1DC}_{2N+1}\rangle|^2 \propto \nonumber\\ 
    &|\langle \bigoplus_{j=0}^N m_{2j+1} | R_X(\tilde{\theta}_{2N+1}) \left[\prod_{j=1}^N R_Z(\tilde{\theta}_{2j})R_X(\tilde{\theta}_{2j-1}) \right] |0\rangle|^2,
    \label{eq:adaptive_MBQC_su2}
\end{align}
where $\boldsymbol{\theta}\in[0,2\pi)^{2N+1}$, $\tilde{\theta}_j = (-1)^{\sum_{k=1}^n A_{jk}m_k}\theta_{j}$, and
\begin{align}
\label{eq:1DC_adaptation_matrix}
    A_{j,k} = 
    \begin{cases}
    1 &\textrm{if }k<j\textrm{ and }k\neq j\textrm{ mod }2 \\
    0 & \textrm{otherwise}
    \end{cases}.
\end{align}
For completeness, we prove this fact in Appendix~\ref{sec:MBQC-1QC_correspondence}.
One can simulate the output of a one-qubit circuit consisting of alternating $X$ and $Z$ rotations if each qubit $j$ is measured adaptively.  In particular, if we wait to measure each qubit $j$ only after all qubits before it (i.e., with label $k<j$) are measured, then we can measure in the eigenbasis of $X(\tilde{\theta}_j)$ instead of $X(\theta_j)$. This allows us to recover any desired one-qubit quantum circuit via the Euler angle decomposition.  The output of the one-qubit circuit is then the parity of all the measurement outcomes for qubits with an odd label.

Adaptive measurement is not always necessary.  Namely, if $\theta_j = n\pi$ for some $n\in\mathbb{Z}$, the factor $(-1)^{\sum_{k=1}^n A_{jk}m_k}$ simply adds a global phase to the circuit, which doesn't influence Eq.~(\ref{eq:adaptive_MBQC_su2}).  This is partly how we will obtain adaptive $l2$-MBQC protocols with $T_C=\Theta(1)$ in our main result. 

\subsection{Implementing QSP circuits with the 1D cluster state}

The alternating pattern of $X$ and $Z$-rotations in Eq.~(\ref{eq:adaptive_MBQC_su2}) suggests that the 1D cluster state is well suited for implementing single-qubit circuits with the QSP decomposition structure, like in Eq.~(\ref{eq:QSP_sym_Boole}).  Since the unitary in Eq.~(\ref{eq:QSP_sym_Boole}) consists of $8n+3$ blocks of mutually commuting unitaries, it can be implemented via adaptive $l2$-MBQC with a 1D cluster state of $O(n^2)$ qubits in time $T_C = O(n)$.  In particular, we have the following theorem.
\begin{theorem}
\label{thm:l2-MBQC_QSP_1DC}
Any symmetric Boolean function $f:\mathbb{F}_2^n\rightarrow \mathbb{F}_2$ can be computed via an adaptive $l2$-MBQC on a 1D cluster state with resource costs $L_C=n$, $L_Q = 8n^2 + 10n +1$, $T_C = 8n + 2$, and $T_Q=3$.
\end{theorem}
\begin{proof}
We simply sketch the proof here. A detailed description of the $l2$-MBQC scheme can be obtained by slightly extending our construction in the proof of Theorem~\ref{thm:mod_p_mbqc}.  

First note that the adaptation matrix in Eq.~\ref{eq:1DC_adaptation_matrix} changes the sign of the measurement basis to $\pm \theta$.  To ensure all measurement settings are binary (i.e., rotation by $\pm \theta$ rather than 0, $+\theta$, or $-\theta$), we can use the fact that $x=(1-(-1)^x)/2$ $\forall x\in\mathbb{F}_2$ to decompose each sequence of $X$-rotations as
\begin{align}
    R_X\left( \frac{\pi |\mathbf{x}|}{n+1} \right) = R_X\left(\frac{\pi}{2(n+1)}\right) \times \\ \prod_{j=1}^n R_X\left((-1)^{x_j+1} \frac{\pi}{2(n+1)}\right). \nonumber
\end{align}
Inserting identities written as $R_Z(0)$ in between each of these $(n+1)$ many $X$-rotations recasts Eq.~(\ref{eq:QSP_sym_Boole}) in a form that matches the right-hand side of Eq.~(\ref{eq:adaptive_MBQC_su2}).  The circuit contains $4n+1$ many sequences of $X$-rotations (each containing $n$ many $R_Z(0)$ rotations sandwiched by $n+1$ $X$-rotations) interspersed by $4n$ nonzero $Z$-rotations.  It follows that the side of the cluster-state required and the number of times we must adapt are $L_Q = (4n+1)(2n+1) + 4n$ and $T_C = 4n+1 + 4n+1$.
\end{proof}

While the above construction is sufficient to compute any symmetric Boolean function (e.g., $\mathsf{Mod}_{p,j}$) using only $\Theta(n^2)$ many qubits initialized in a 1D cluster state, it has run time $\Theta(n)$.  Meanwhile, the constructions in Sec.~\ref{sec:OR_Red_MBQC} have constant run time.  In the next section, we will show how one can extend the utility of the QSP approach to acheive constant-time $l2$-MBQCs.

\subsection{Simplified linear-space constant-time $l2$-MBQC for $\mathsf{Mod}_{3,0}$}
\label{sec:adaptive_n_n}

\begin{figure*}
    \centering
    \includegraphics[width=\linewidth]{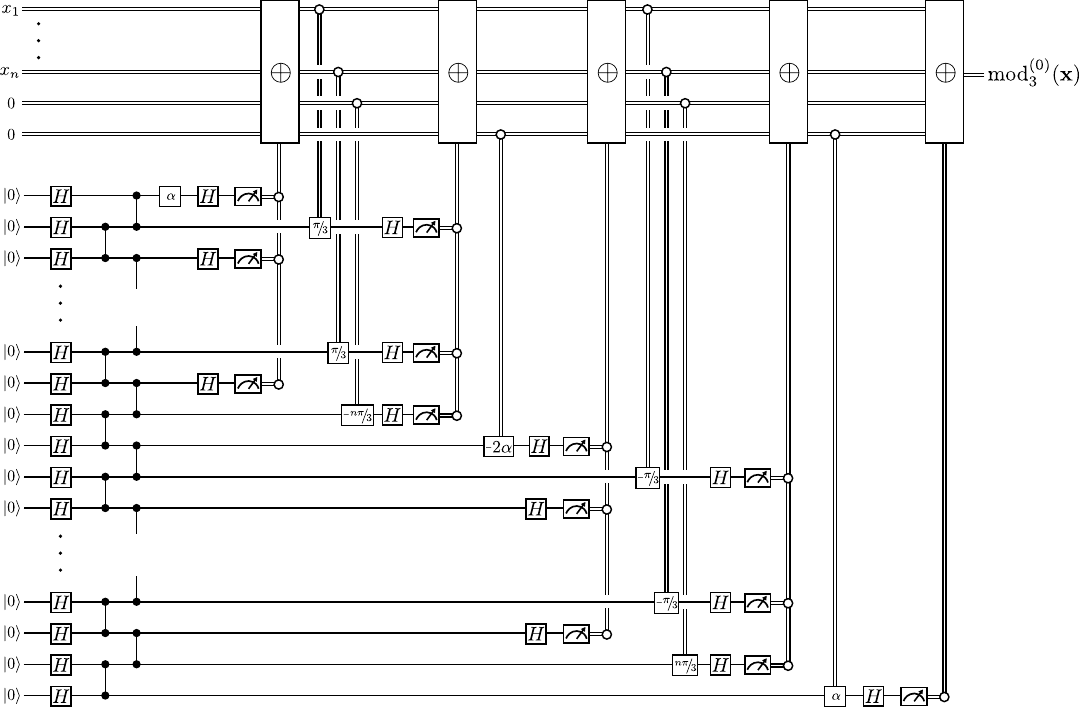}
    \caption{A sketch of the adaptive $l2$-MBQC protocol that computes the function $\mathsf{Mod}_{3,0}:\mathbb{F}_2^n\rightarrow\mathbb{F}_2$.  The scheme uses $n+2$ classical bits, the first $n$ of which are the string $\mathbf{x}$.  The other two bits carry the parity of all measurement outcomes obtained at that point for the even and odd qubits, respectively.  Each classically conditioned single-qubit gate labeled by an angle $\theta$ represents a $R_Z(\pm\theta)$ gate where the sign is determined by the bit it is conditioned on. The scheme proceeds in five rounds of interaction with the mod-2 linear classical side-processor and gives an MBQC realization of the circuit in Eq.~(\ref{eq:Mod3_Circuit}).}
    \label{fig:mod_3_circuit}
\end{figure*}

We begin with a particularly intuitive 1QC that computes the mod-3 function.  This simplified construction does not require us to resort to the machinery of quantum signal processing.
\begin{lemma}
\label{lem:nc1QCmod_3}
The function $\mathsf{Mod}_{3,0}:\mathbb{F}_2^n\rightarrow\mathbb{F}_2$ can be computed by a 1QC with $L_C=n$ and $T_Q=2n+1$.
\end{lemma}
\begin{proof}
The proof is constructive.  Let
\begin{align}
U_Q = \exp\left(i\frac{\pi}{3} \frac{X+Y+Z}{\sqrt{3}}\right).
\end{align}
This adjoint action of this unitary on Pauli operators gives the following map.
\begin{align}
    U_Q\odot U_Q^\dagger = 
    \begin{cases}
    X\mapsto Z \\
    Y\mapsto X \\
    Z\mapsto Y 
    \end{cases},
\end{align}
where the $\odot$ denotes suspended evaluation.  Consider now the following sequence of unitary transformations
\begin{align}
    U(\mathbf{x}) = \left[\prod_{j=1}^nU_Q^{x_j}\right] Z \left[\prod_{j=1}^n(U_Q^{x_j})^\dagger\right].
    \label{eq:Mod3_Circuit}
\end{align}
It then follows that
\begin{align}
U(\mathbf{x}) = 
    \begin{cases}
    Z &\textrm{if } |\mathbf{x}| = 0 \textrm{ mod }3\\
    Y &\textrm{if }|\mathbf{x}| = 1 \textrm{ mod }3\\
    X &\textrm{if }|\mathbf{x}| = 2 \textrm{ mod }3
    \end{cases}
\end{align}
and hence
\begin{align}
    U(\mathbf{x}) |0\rangle \propto |\mathrm{mod}_3^{(0 )}(\mathbf{x})\rangle.
\end{align}
Thus, $U(\mathbf{x})$ realizes a 1QC of the mod-3 function with $L_C=n$ and $T_C=2n+1$.
\end{proof}

We remark that the circuit in Eq.~(\ref{eq:Mod3_Circuit}) consists of Clifford gates only and has potential to be implemented in a fault-tolerant setting \cite{bravyi2020quantum} using a method like gate teleportation as in Ref.~\cite{caha2022single}. 

Notice that the circuit $U(\mathbf{x})$ consists of three blocks of mutually commuting unitaries (i.e., the first $n$ gates $U_Q^\dagger$, the middle $Z$ gate, and the last $n$ gates $U_Q$).  It follows that one can realize an adaptive $l2$-MBQC protocol with this circuit as the underlying 1QC that only uses a constant number of rounds of interactivity with the classical side processor.  This $l2$-MBQC scheme, which we describe in the following theorem, is sketched in Fig.~\ref{fig:mod_3_circuit}.
\begin{theorem}
The function $\mathsf{Mod}_{3,0}:\mathbb{F}_2^n\rightarrow\mathbb{F}_2$ can be computed via an adaptive $l2$-MBQC on a 1D cluster state with resource costs $L_C=n+2$, $L_Q=4n+5$, $T_C = 5$, and $T_Q=3$.
\label{thm:aMQC_mod3}
\end{theorem}
\begin{proof}
We begin by rewriting the circuit in Eq.~(\ref{eq:Mod3_Circuit}) in terms of $X$ and $Z$ rotations to match Eq.~(\ref{eq:adaptive_MBQC_su2}), obtaining $U_Q = R_Z(\beta)R_X(\alpha) R_Z(2\pi/3) R_X(\alpha)^\dagger R_Z(\beta)^\dagger$ were $\cos(\alpha)=-1/\sqrt{3}$, $\sin(\alpha)=\sqrt{2/3}$, and $\beta=-\pi/4$. The first and last $R_Z(\beta)$ rotations simply add a global phase and so $|\langle m | U(\mathbf{x}) |0\rangle|^2 = |\langle m | V(\mathbf{x}) |0\rangle|^2$  where
\begin{align}
    V(\mathbf{x}) = &R_X(\alpha) R_Z\left(\frac{2\pi|\mathbf{x}|}{3}\right)R_X(2\alpha)^\dagger \times \nonumber\\
    &R_Z\left(\frac{2\pi|\mathbf{x}|}{3}\right)^\dagger R_X(\alpha).
\end{align}

We further decompose the Pauli-$Z$ rotations in $V(\mathbf{x})$ as
\begin{align}
    R_Z\left(\frac{2\pi |\mathbf{x}|}{3}\right) = R_Z\left(\frac{n\pi}{3}\right)\prod_{j=1}^n R_Z\left((-1)^{x_j+1}\frac{\pi}{3}\right).
\end{align}

Now consider a $(4n+5)$-qubit 1D cluster state. The the above circuit can be implemented in a measurement-based manner via the following sequence of measurements.
\begin{itemize}
    \item[(1)] Qubit 1 is measured in the eigenbasis of $X(\alpha)$.  Meanwhile, each qubit with label $j=1\textrm{ mod }2$ and $1<j\leq 2n+1$ is measured in the Pauli-$X$ basis.  Each measurement outcome $m_j\in\mathbb{F}_2$ is returned to the mod-2 linear classical side-processor.
    \item[(2)] For each qubit with label $j = 0 \textrm{ mod }2$ and $1<j\leq 2n$, the side-processor computes and returns the bit
    \begin{align}
        s_j = \left(\bigoplus_{\substack{k<j \\ k=1\textrm{ mod }2}} m_k\right) \oplus x_{(j/2)}.
    \end{align}
    Each qubit is then measured in the eigenbasis of $X(\theta_j)$ where $\theta_j = (-1)^{s_j}\pi/3$. Meanwhile, for qubit $2n+2$ the side-processor computes and returns the bit
    \begin{align}
        s_{2n+2} =  \bigoplus_{\substack{k<2n+2\\k=1\textrm{ mod }2}} m_k.
    \end{align}
    The qubit is then measured in the eigenbasis of $X(\theta_{2n+2})$ where $\theta_{2n+2} = (-1)^{s_{2n+2}+1}n\pi/3$.  Each measurement outcome $m_j\in\mathbb{F}_2$ is returned to the side-processor.
    \item[(3)] For qubit $2n+3$, the side-processor computes and returns the bit
    \begin{align}
        s_{2n+3} = \bigoplus_{\substack{k<2n+3 \\ k = 0\textrm{ mod }2}} m_k.
    \end{align}
    The qubit is then measured in the eigenbasis of $X(\theta_{2n+3})$ where $\theta_{2n+3} = (-1)^{s_{2n+3}+1} 2\alpha$. Meanwhile, each qubit with label $j=1\textrm{ mod }2$ and $2n+3<j\leq 4n+3$ is measured in the Pauli-$X$ basis. Each measurement outcome $m_{j}\in\mathbb{F}_2$ is returned to the side-processor.
    \item[(4)] For each qubit with label $j=0\textrm{ mod }2$ and $2n+2<j\leq 4n+2$, the side-processor computes and returns the bit
    \begin{align}
        s_j = \left(\bigoplus_{\substack{k<j \\ k=1\textrm{ mod }2}} m_k\right) \oplus x_{(j-2n-2)/2}.
    \end{align}
    The qubit then is measured in the eigenbasis of $X(\theta_j)$ where $\theta_j = (-1)^{s_j+1}\pi/3$.  Meanwhile, for qubit $4n+4$ the side-processor computes and returns the bit
    \begin{align}
        s_{4n+4} = \left(\bigoplus_{\substack{k<4n+4 \\ k=1\textrm{ mod }2}} m_k\right).
    \end{align}
    The qubit is then measured in the eigenbasis of $X(\theta_{4n+4})$ where $\theta_{4n+4} = (-1)^{s_{4n+4}}n\pi/3$.  Each measurement outcome $m_j\in\mathbb{F}_2$ is returned to the side-processor.
    \item[(5)] For qubit $4n+5$ the side-processor computes and returns the bit
    \begin{align}
        s_{4n+5} = \bigoplus_{\substack{k< 4n+5 \\ k= 0\textrm{ mod }2}} m_k.
    \end{align}
    The qubit is then measured in the eigenbasis of $X(\theta_{4n+5})$ where $\theta_{4n+5} = (-1)^{s_{4n+5}}\alpha$.  The measurement outcome $m_{4n+5}\in\mathbb{F}_2$ is returned to the side-processor, which then computes and returns the computational output
    \begin{align}
        y = \bigoplus_{j = 1\textrm{ mod }2} m_j.
    \end{align}
\end{itemize}
It follows from Eq.~(\ref{eq:adaptive_MBQC_su2}) that $y=\mathsf{Mod}_{3,0}(\mathbf{x})$.
\end{proof}

\subsection{Linear-space constant-time $l2$-MBQC for $\mathsf{Mod}_{p,j}$}

The entire family of mod-$p$ functions defined in Eq.~(\ref{eq:mod_p_def}) can be computed via the quantum signal processing technique used in Ref.~\cite{maslov2020quantum}.  The following theorem improves on their construction for mod-$p$ functions, giving a quadratic reduction in the required resources and enabling constant-time adaptive $l2$-MBQCs.  This quadratic improvement is due to the fact that the truth table of the mod-$p$ functions is invariant to shifting the Hamming weight by $p$, so we only need to look at a piece of the truth table rather than the whole thing.  Giving a circuit with $O(p)$ many rotations by the Hamming weight.

\begin{lemma}
\label{lem:nc1QCmod_p}
For any odd number $p$ and $j\in\{0,\ldots,p-1\}$, the function $\mathsf{Mod}_{p,j}:\mathbb{F}_2^n\rightarrow\mathbb{F}_2$ can be computed by a 1QC with $L_C=n$ and $T_Q=(2p-1)n + 2p$.
\end{lemma}
\begin{proof}
In Appendix~\ref{sec:Previous_1QC_Work} we demonstrate that there exist a set of $2p-1$ angles $\{\xi_j\in[0,2\pi)\}_{j=1}^{2p-1}$ such that the unitary
\begin{align}
    U(\mathbf{x}) = \prod_{j=1}^{2p-1} R_Z(\xi_j)  R_X\left(\frac{4\pi |\mathbf{x}|}{p}\right) R_Z(\xi_j)^\dagger
    \label{eq:mod_p_qsp_circuit}
\end{align}
satisfies $|\langle y |U(\mathbf{x}) | 0\rangle|^2 = \delta_y^{\mathsf{Mod}_{p,j}(\mathbf{x})}$.  We also explain in detail an algorithm to determine these angles numerically.  These angles are listed in Tab.~\ref{tab:mod_p_qsp_angles} for $j=0$ and $p<9$.
\end{proof}

\begin{table}[]
    \centering
    \begin{tabular}{c||c|c|c|c}
        $p$ & 3 & 5 & 7 & 9 \\
        \hline
        $\xi_1$ & -0.21032 & 0.25795 & 0.24598 & -1.32875 \\
        $\xi_2$ & 0.62099 & 0.08709 & 0.21709 & -0.79511 \\
        $\xi_3$ & 2.64302 & -0.47767 & 0.00603 & -0.10787 \\
        $\xi_4$ & 1.75347 & -1.55500 & -0.42033 & 0.88376 \\
        $\xi_5$ & 2.39109 & 2.89580 & -1.11688 & 3.0817 \\
        $\xi_6$ & - & -1.78858 & -2.14572 & 1.53016 \\
        $\xi_7$ & - & -1.80615 & 2.37267 & 1.07858 \\
        $\xi_8$ & - & -2.17667 & -2.04731 & 1.15532 \\
        $\xi_9$ & - & -2.64310 & -1.72877 & 1.68658 \\
        $\xi_{10}$ & - & - & -1.82919 & 2.26212 \\
        $\xi_{11}$ & - & - & -2.08722 & 2.61176 \\
        $\xi_{12}$ & - & - & -2.41079 & 3.02345 \\
        $\xi_{13}$ & - & - & -2.75310 & -2.32569 \\
        $\xi_{14}$ & - & - & - & 1.54469 \\
        $\xi_{15}$ & - & - & - & 1.19266 \\
        $\xi_{16}$ & - & - & - & 1.26558 \\
        $\xi_{17}$ & - & - & - & 1.45910 
    \end{tabular}
    \caption{A table of the angles $\{\xi_j\}_{j=1}^{2p-1}$ in radians used in the quantum signal processing approach to computing the function $\mathrm{mod}_p^{(0)}:\mathbb{F}_2^n\rightarrow\mathbb{F}_2$ for the first few nontrivial values of $p$.  These angles were determined numerically using the algorithm described in Appendix~\ref{sec:Previous_1QC_Work} and are thus approximations of the true angles required.  In the worst case, the probability of failure for a 1QC using these finite decimal approximations is less than $10^{-10}$.  We also remark that for each additional least significant digit thrown out the probability of failure increases by a factor of $10^{2}$.}
    \label{tab:mod_p_qsp_angles}
\end{table}

Recall from the proof of Theorem~\ref{thm:l2-MBQC_QSP_1DC} that the number of times the mod-2 linear classical side-processor is called is twice the number of $X$-rotations by the Hamming weight that appear in the circuit.  In Eq.~(\ref{eq:mod_p_qsp_circuit}) this number is $2(2p-1)$, which is independent of $n$. Hence, the unitary in Eq.~(\ref{eq:mod_p_qsp_circuit}) can be implemented via adaptive $l2$-MBQC with only a constant number of rounds of interaction with the mod-2 linear classical side-processor.  This $l2$-MBQC scheme is described in the following theorem.
\begin{theorem}
For any odd number $p$, the function $\mathsf{Mod}_{p,j}:\mathbb{F}_2^n\rightarrow\mathbb{F}_2$ can be computed via an adaptive $l2$-MBQC on a 1D cluster state with resource costs $L_C=n+2$, $L_Q= (4p-2)(n+1) - 1$, $T_C=4p-2$, and $T_Q=3$.
\label{thm:mod_p_mbqc}
\end{theorem}
\begin{proof}
We begin by noticing that the first and last $Z$-rotations in the circuit in Eq.~(\ref{eq:mod_p_qsp_circuit}) do not change $|\langle y|U(\mathbf{x})|0\rangle|^2$ and thus do not affect the probability for $y$ to be 0 or 1. Call this unitary obtained from excluding these rotations $V(\mathbf{x})$, i.e., $V(\mathbf{x}) = R_Z(\xi_{2p-1})^\dagger U(\mathbf{x}) R_Z(\xi_1)$.  By further decomposing each $X$-rotation as
\begin{align}
     R_X\left(\frac{4\pi |\mathbf{x}|}{p}\right) = &R_X\left( \frac{2\pi n}{p} \right)\times \\
    &\prod_{k=1}^n R_X\left( (-1)^{x_k+1}\frac{2\pi}{p} \right),\nonumber
\end{align}
we may implement the circuit $V(\mathbf{x})$ on a $\{(4p-2)(n+1) - 1\}$-qubit cluster state via the following sequence of adaptive measurements.
\begin{itemize}
    \item[(1)] Each qubit with label $j=0\textrm{ mod }2$, except those with label $j=k(2n+2)$ for $k\in[2p-2]$, are measured in the Pauli-$X$ basis.  Each measurement outcome $m_j\in\mathbb{F}_2$ is returned to the mod-2 linear classical side-processor.
    \item[] The next $4p-4$ measurements proceed as follows for each $\mu\in[2p-2]$.
    \begin{itemize}
        \item[($2\mu$)] For each qubit with label $j=1\mod2$ and $(\mu-1)(2n+2)+1\leq j<\mu(2n+2)-1$ the side-processor computes and returns the bit 
        \begin{align}
            s_j = \left(\bigoplus_{\substack{k<j \\ k = 0\textrm{ mod }2}}m_k\right)\oplus x_{\iota_\mu(j)},
        \end{align}
        where $\iota_\mu(j) = (j+1)/2 - (\mu-1)(n+1)$ is a map denoting which classical bit $x_i$ each qubit $j$ is conditioned on.  Each qubit is then measured in the eigenbasis of $X(\theta_j)$ where $\theta_j = (-1)^{s_j+1}2\pi/p$.  Meanwhile, for qubit $\mu(2n+2)-1$ the side-processor computes and returns the bit
        \begin{align}
             s_{\mu(2n+2)-1} = \bigoplus_{\substack{k<\mu(2n+2)-1 \\ k = 0\textrm{ mod }2}}m_k.
        \end{align}
        The qubit is then measured in the eigenbasis of $X(\theta_{\mu(2n+2)-1})$ where $\theta_{\mu(2n+2)-1} = (-1)^{s_{\mu(2n+2)-1}}2\pi n/p$.
        Each measurement outcome $m_j\in\mathbb{F}_2$ is returned to the side-processor.
        \item[($2\mu + 1$)] For qubit $\mu(2n+2)$ the side-processor computes and returns the bit
        \begin{align}
             s_{\mu(2n+2)} = \bigoplus_{\substack{k<\mu(2n+2) \\ k = 1\textrm{ mod }2}}m_k.
        \end{align}
        The qubit is then measured in the eigenbasis of $X(\theta_{\mu(2n+2)})$ where $\theta_{\mu(2n+2)} = (-1)^{s_{\mu(2n+2)}}(\xi_{\mu+1} - \xi_\mu)$.  The measurement outcome $m_{\mu(2n+2)}\in\mathbb{F}_2$ is returned to the side-processor.
    \end{itemize}
    \item[($4p-2$)] For each qubit with label $j=1\mod2$ and $(4p-4)(n+1)+1\leq j<(4p-2)(n+1)-1$ the side-processor returns the bit 
        \begin{align}
            s_j = \left(\bigoplus_{\substack{k<j \\ k = 0\textrm{ mod }2}}m_k\right)\oplus x_{\iota(j)},
        \end{align}
        where $\iota(j) = (j+1)/2 - (2p-2)(n+1)$.  Each qubit is then measured in the eigenbasis of $X(\theta_j)$ where $\theta_j = (-1)^{s_j+1}2\pi/p$.  Meanwhile, for qubit $(4p-2)(n+1)-1$ the side-processor returns the bit
        \begin{align}
             s_{(4p-2)(n+1)-1} = \bigoplus_{\substack{k<(4p-2)(n+1)-1 \\ k = 0\textrm{ mod }2}}m_k.
        \end{align}
        The qubit is measured in the eigenbasis of $X(\theta_{(4p-2)(n+1)-1})$ where $\theta_{(4p-2)(n+1)-1} = (-1)^{s_{(4p-2)(n+1)-1}}2\pi n/p$.
        Each measurement outcome $m_j\in\mathbb{F}_2$ is returned to the side-processor, which then computes and returns the computational output
        \begin{align}
            y = \bigoplus_{j=1\textrm{ mod }2} m_j.
        \end{align}
\end{itemize}
It follows from Eq.~\ref{eq:adaptive_MBQC_su2} that $y=\mathsf{Mod}_{p,j}(\mathbf{x})$.
\end{proof}

\section{Quantum advantages via adaptive $l2$-MBQC}
\label{sec:classical_comparison}

The adaptive $l2$-MBQCs we have considered are capable of outperforming various classical computational models. The weakest such classical models are the noncontextual hidden variable models \cite{raussendorf2013contextuality, abramsky2017contextual} (denoted $\mathsf{NCHVM}$), which have the same the power as the mod-2 linear classical side-processor alone \cite{hoban2011generalized}.  The probability for the classical side-processor to compute a Boolean function $f:\mathbb{F}_2^n\rightarrow\mathbb{F}_2$ on an input $\mathbf{x}\in\mathbb{F}_2^n$ chosen uniformly at random can be written as a positive sum of conditional probabilities that is bounded as,
\begin{align}
        \frac{1}{2^n}\sum_{\mathbf{x}\in\mathbb{F}_2^n}\sum_{\substack{\mathbf{m}\in\mathbb{F}_2^{L_Q}\\\mathbf{o}^T\mathbf{m}=f(\mathbf{x})}} \mathrm{pr}(\mathbf{m}|\mathbf{s}=P\mathbf{x} + A\mathbf{m})\leq_{\mathsf{NCHVM}}\beta.
        \label{eq:adaptive_bell}
\end{align}
By Fine's theorem \cite{fine1982hidden}, this bound is proportional to the Hamming distance between $f$ and the closest linear Boolean function.   Eq.~(\ref{eq:walsh_hadamard_def}) thus implies that,
\begin{align}
    \beta= \frac{1 + \hat{f}_\mathrm{max}}{2}.
\end{align}
In the nonadaptive setting, Eq.~(\ref{eq:adaptive_bell}) corresponds to a so-called $(n,2,2)$ Bell scenario, which have been fully characterized in Ref.~\cite{werner2001all}.  Meanwhile, our adaptive $l2$-MBQC protocols give examples of Bell inequalities with an additional spatio-temporal ordering that lie beyond this characterization.

While a violation of the Bell inequality in Eq.~(\ref{eq:adaptive_bell}) shows that the $l2$-MBQC is in some sense nonclassical, Theorems~\ref{thm:aMQC_mod3} and \ref{thm:mod_p_mbqc} indicate a stronger separation in the context of circuit complexity that was first shown in Ref.~\cite{moore1999quantum}.  To understand the content of this theorem, we first define a few terms.  Let $\mathsf{QNC}^0$ denote the class of polynomial-size constant-depth quantum circuits consisting of gates with bounded fan-in.  Furthermore, let $\mathsf{QNC}^0[2]$ denote the same class of circuits, but with the assistance of a mod-2 linear classical side-processor (i.e., with oracle access to the function $\mathsf{Mod}_{2,0}:\mathbb{F}_2^m \rightarrow \mathbb{F}_2$ for arbitrarily large $m$) such that $T_Q+T_C$ is some constant, independent of $n$. Let $\mathsf{AC}^0$ denote the class of polynomial-size constant-depth classical circuits consisting of AND, OR, and NOT gates with unbounded fan-in and fan-out.  Furthermore, let $\mathsf{AC}^0[p]$ denote the same class of circuits, but with the oracle access to the function $\mathsf{Mod}_{p,0}:\mathbb{F}_2^m \rightarrow \mathbb{F}_2$ for arbitrarily large $m$.

All $l2$-MBQC algorithms discussed in this paper violate the inequality in Eq.~(\ref{eq:adaptive_bell}) and thus constitute a proof of the rather trivial separation $\mathsf{NCHVM}\subset \mathsf{QNC}^0[2]$ (c.f., Tab.~\ref{tab:resource_cost_table}). However, our construction in Sec.~\ref{sec:adaptive_n_n} can be used to prove a more powerful separation that concerns the notions of $\mathsf{AC}^0$-reducibility and $\mathsf{QNC}^0$-reducibility. A Boolean function $f$ is said to be $\mathsf{AC}^0$-reducible (respectively, $\mathsf{QNC}^0$-reducibile) to another function $g$ if a $\mathsf{AC}^0$ circuit (respectively, a $\mathsf{QNC}^0$ circuit) with oracular access to $g$ can compute $f$. Thm.~2 of Ref.~\cite{smolensky1987algebraic} showed that for every odd prime number $p$, $\mathsf{Mod}_{p,0}$ is not $\mathsf{AC}^0$-reducible to $\mathsf{Mod}_{2,0}$.  However, our construction in Thm.~\ref{thm:mod_p_mbqc} gives a $\mathsf{QNC}^0$ reduction.  Therefore, we have the following corollary. 
\begin{corollary}
For every prime number $p$ the function $\mathsf{Mod}_{p,0}$ is $\mathsf{QNC}^0$-reducible to $\mathsf{Mod}_{2,0}$.
\end{corollary}
\noindent Furthermore, Ref.~\cite{smolensky1987algebraic} also showed that for any two primes $p$ and $r$ with $p<r$, $\mathsf{Mod}_{r,0}$ is not $\mathsf{AC}^0$-reducible to $\mathsf{Mod}_{p,0}$.  However, every prime number $r$ is $\mathsf{QNC}^0$-reducible to $\mathsf{Mod}_{2,0}$.  Hence, we have the following theorem.
\begin{corollary}
For every prime number $p$, $\mathsf{QNC}^0[2]\not\subset \mathsf{AC}^0[p]$.
\end{corollary}
\begin{proof}
For any prime number $p$ consider the next prime number $r$ with $p<r$.  Since $\mathsf{Mod}_{r,0}$ is not $\mathsf{AC}^0$-reducible to $\mathsf{Mod}_{p,0}$, it lies outside the class $\mathsf{AC}^0[p]$.  Therefore, for any $p$ there is always a function that lives in the class $\mathsf{QNC}^0[2]$ but not in $\mathsf{AC}^0[p]$.  Hence, $\mathsf{QNC}^0[2]\not\subset \mathsf{AC}^0[p]$.
\end{proof}
\noindent The above separation was originally proven in Ref.~\cite{moore1999quantum}.  As shown in Tab.~\ref{tab:resource_cost_table}, our construction gives an overall better scaling in the number of mid-circuit measurements required, namely $\Theta(p)$ versus $O(p^2\log^3(p))$.  We also remark that  Ref.~\cite{takahashi2016collapse} demonstrates the more powerful separation $\mathsf{ACC}^0\subseteq \mathsf{QNC}^0[2]$ where $\mathsf{ACC}^0 = \cup_{p} \mathsf{AC}^0[p]$.

Recently, nonadaptive $l2$-MBQCs that use a cluster state resource and no post processing have been leveraged to demonstrate separations between the class of constant-depth quantum circuits (denoted $\mathsf{QNC}^0$) and constant-depth classical circuits of gates with unbounded fan-out and either bounded fan-in ($\mathsf{NC}^0$ \cite{bravyi2018quantum}) or unbounded fan-in ($\mathsf{AC}^0$ \cite{watts2019exponential, grier2019interactive}).  It remains an open problem to show whether or not $\mathsf{QNC}^0\not\subset\mathsf{AC}^0[p]$.  While Ref.~\cite{grier2019interactive} showed $\mathsf{QNC}^0\not\subset\mathsf{AC}^0[6]$, our techniques may be useful in going beyond this.  A first step might be to understand whether our construction in Thm.~\ref{thm:aMQC_mod3} can give a strategy for a game akin to the parity bending problem of \cite{watts2019exponential} when we remove the mod-2 linear side-processor.

\section{Conclusion and outlook}
\label{sec:conclusion_outlook}

As summarized in the flow chart of Fig~\ref{fig:flow_chart}, we have shown how adaptive $l2$-MBQC with 1D cluster states leads to exponential reductions in the space-time resources required to compute symmetric Boolean functions.  Namely, we showed how the family of mod-$p$ functions can be computed using a linear-size 1D cluster state with only a constant number of queries to a mod-2 linear classical side-processor.  Our constructions are based on an extension of the quantum signal processing technique used in Ref.~\cite{maslov2020quantum} to compute Boolean functions via 1QC.  While our $l2$-MBQCs are guaranteed to be contextual, violating a corresponding noncontextuality inequality in Eq.~(\ref{eq:adaptive_bell}), we also sketch how our results recover the stronger separation against constant-depth classical circuits  $\mathsf{QNC}^0[2]\not\subset\mathsf{AC}^0[p]$ for every prime number $p$, which was first shown in Ref.~\cite{moore1999quantum}.   We also remark that this work sheds light on a number of new examples of contextual quantum computations with nontrivial temporal order that may be of interest in further developing the cohomological framework in \cite{raussendorf2022putting}.  Though all the algorithms we have studied here are deterministic, a detailed analysis of the bounded error case \cite{hoyer2005quantum, takahashi2016collapse, mori2018periodic} may give further insight to the role of contextuality in these constructions.  In particular, it would be interesting to see if the known algorithms saturate the upper bound on the success probability for an $l2$-MBQC given by the so-called contextual fraction \cite{abramsky2017contextual, okay2018cohomological}.

In terms of an interesting near-term implementation, we remark that our construction in Thm.~\ref{thm:mod_p_mbqc} gives the best known scaling for the number of resources required to compute the mod-$p$ functions.  Several recent works have used small nonadaptive $l2$-MBQCs of Boolean functions to benchmark near-term devices \cite{demirel2021correlations, sheffer2022playing, daniel2022quantum}.  It may be interesting to implement our simplest nontrivial example (i.e., the function $\mathsf{Mod}_{3,0}:\mathbb{F}_2^4\rightarrow\mathbb{F}_2$ via the algorithm in Theorem~\ref{thm:aMQC_mod3}) to characterize the quality of a device capable of performing mid-circuit measurements. Our protocol using cluster states only requires 21 qubits whereas a quick calculation shows that the periodic Fourier decomposition, the OR-reduction, and Moore's counting circuit require 29, 34, and 40 qubits, respectively. Furthermore, while our constant-depth implementation of $\mathsf{Mod}_{p,0}:\mathbb{F}_2^n\rightarrow\mathbb{F}_2$ has better resource scaling than previous methods (c.f. Tab.~\ref{tab:resource_cost_table}), we do not know if it is optimal.  We wonder if there a figure of merit similar to the periodic Fourier sparsity that bounds the number of conditional $\mathrm{U}(2)$ transformations required to represent a particular Boolean function. Perhaps ideas from optimal synthesis of single-qubit unitaries may be of use \cite{kliuchnikov2013fast, forest2015exact}.

Another interesting direction is developing an understanding of resource states for $l2$-MBQC from the perspective of condensed matter physics. Ground states of many-body Hamiltonians with 1D symmetry-protected topological order (a quantum phase of matter having similar properties to the cluster state) have been shown to be useful for a variety of tasks in $l2$-MBQC \cite{daniel2021quantum}.  A similar advantage was extended to a whole ferromagnetic phase in Ref.~\cite{bulchandani2022playing}.  We wonder if ground states of a quantum phase of matter might have uniform capability to perform adaptive $l2$-MBQC tasks.  It has been noted that the ability of 1D SPTO ground states to be converted to ``GHZ-like" states via local measurements, a key property for a measurement-based implementation of quantum fan-out, is a general property of 1D SPTO phases \cite{tantivasadakarn2021long, verresen2021efficiently}.

\acknowledgments
We thank Alex Fischer and Benjamin Goff for reading the manuscript and providing feedback.
This work was supported partially by the National Science Foundation STAQ Project (PHY-1818914), PHY-1915011.  Support is also acknowledged from the U.S. Department of Energy, Office of Science, National Quantum Information Science Research Centers, Quantum Systems Accelerator.

\begin{widetext}

\appendix

\section{Details on Boolean functions}
\label{sec:boolean_appendix}

\subsection{Algebraic normal forms for the mod-$p$ functions}
\label{sec:mod3ANF}

The algebraic normal form of any symmetric Boolean function can be expressed as a combination of the complete $\mu$-tic functions $C_n^\mu(\mathbf{x})$, which have the ANFs
\begin{align}
    C_n^\mu(\mathbf{x}) = \bigoplus_{\substack{S\subseteq [n] \\ |S|=\mu}} \prod_{j\in S} x_j.
\end{align}
In terms of these functions, we have the following.
\begin{lemma}
The ANFs for the $\mathrm{mod}$-3 functions are
\begin{align}
    \mathsf{Mod}_{3,j}(\mathbf{x}) &= \bigoplus_{\substack{\mu=0 \\ \mu \neq 3-j \textrm{ }\mathrm{mod}\textrm{ }3}}^n C_n^\mu(\mathbf{x}). \label{eq:mod_3_anf_j}
\end{align}
\end{lemma}
\begin{proof}
The proof works by induction.  As a base case it is easily checked that $\mathsf{Mod}_{3,0}(x_1,x_2,x_3) = x_1 \oplus x_2 \oplus x_3 \oplus x_1x_2 \oplus x_1x_3 \oplus x_2 x_3$, $\mathsf{Mod}_{3,1}(x_1,x_2,x_3) = 1\oplus x_1 \oplus x_2 \oplus x_3 \oplus x_1x_2x_3$, and $\mathsf{Mod}_{3,2}(x_1,x_2,x_3) = 1\oplus x_1x_2 \oplus x_1x_3 \oplus x_2 x_3 \oplus x_1x_2x_3$.  For the induction step we will make use of the fact that
\begin{align}
\mathsf{Mod}_{3,j} (x_1,\ldots,x_{n+1}) &= (1\oplus x_{n+1}) \mathsf{Mod}_{3,j}(x_1,\ldots,x_n) \oplus x_{n+1} \mathsf{Mod}_{3,j-1}(x_1,\ldots,x_n).
\end{align}
In the above, the superscript is taken modulo 3.
Upon substituting the corresponding ANFs into the above equations, we obtain
\begin{align}
\mathsf{Mod}_{3,j} (x_1,\ldots,x_{n+1}) &= (1\oplus x_{n+1})  \bigoplus_{\substack{\mu=0 \\ \mu \neq 3-j \textrm{ mod }3}}^n C_n^\mu(\mathbf{x}) \oplus x_{n+1}  \bigoplus_{\substack{\mu=0 \\ \mu \neq 1-j \textrm{ mod }3}}^n C_n^\mu(\mathbf{x}).
\end{align}
Expanding the above equations by collecting terms in powers of $x_{n+1}$ gives
\begin{align}
\mathsf{Mod}_{3,j} (x_1,\ldots,x_{n+1}) &=   \bigoplus_{\substack{\mu=0 \\ \mu \neq 3-j \textrm{ mod }3}}^n C_n^\mu(\mathbf{x}) \oplus x_{n+1}  \bigoplus_{\substack{\mu=0 \\ \mu \neq 2-j \textrm{ mod }3}}^n C_n^\mu(\mathbf{x}).
\end{align}
We can reindex the second sum in each equation to sum over a variable $\nu = \mu + 1$, giving
\begin{align}
\mathsf{Mod}_{3,j} (x_1,\ldots,x_{n+1}) &=   \bigoplus_{\substack{\mu=0 \\ \mu \neq 3-j \textrm{ mod }3}}^n C_n^\mu(\mathbf{x}) \oplus x_{n+1}  \bigoplus_{\substack{\nu=0 \\ \nu \neq 3-j \textrm{ mod }3}}^n C_n^{\nu-1}(\mathbf{x}).
\end{align}
Finally, using the fact that
\begin{align}
C_{n+1}^\mu (x_1,\ldots,x_{n+1}) = C_n^{\mu}(x_1,\ldots,x_n)\oplus x_{n+1}C_{n}^{\mu-1}(x_1,\ldots,x_n),
\end{align}
the terms in each sum can be combined to give
\begin{align}
\mathsf{Mod}_{3,j} (x_1,\ldots,x_{n+1}) &= \sum_{\substack{\mu=0 \\ \mu \neq 3-j \textrm{ mod }3}}^{n+1} C_{n+1}^\mu(x_1,\ldots,x_{n+1}).
\end{align}
Thus, by induction Eq.~(\ref{eq:mod_3_anf_j}) holds.
\end{proof}

\begin{lemma}
For each positive integer $n$ and prime number $p$ at least one of the functions $\mathsf{Mod}_{p,j}:\mathbb{F}_2^n\rightarrow\mathbb{F}_2$ has degree-$n$ ANF.
\end{lemma}
\begin{proof}
Taking the ansatz that the ANF for the mod-$p$ functions is
\begin{align}
    \mathsf{Mod}_{p,j}(\mathbf{x}) = \bigoplus_{\mu=0}^n a_\mu^{(j)} C_n^\mu(\mathbf{x})
\end{align}
where each $a_\mu^{(j)}\in\mathbb{F}_2$ and applying the same steps as in the above induction proof we obtain the relation
\begin{align}
    a_{\mu+1}^{(j)} = a_{\mu}^{(j)} \oplus a_{\mu}^{(j-1)} 
\end{align}
where the superscripts are taken modulo $p$.  Arranging the coefficients $a_\mu^{(j)}$ into a length-$p$ binary vector for each $\mu$, denoted $\mathbf{a}_\mu$ we obtain the update rule $\mathbf{a}_\mu = M \mathbf{a}_{\mu-1}$ where $M = I^{(p)} \oplus X^{(p)}$ where $I^{(p)}$ is the $p\times p$ identity matrix and is the $p\times p$ generalized Pauli-$X$ matrix (i.e., $X^{(p)} = \sum_{j=0}^{p-1} |j+1\rangle \langle j|$).  Since the base case is given by $\mathbf{a}_{0} = (0,1,1,\ldots,1)^T$, one can explicitly calculate the coefficients in general.  They are
\begin{align}
    a_\mu^{(j)} = 2^{\mu} \left[1 - \frac{1}{p} \sum_{k=0}^{p-1} \cos^{\mu}\left(\frac{k\pi}{p} \right) e^{i\frac{k\pi}{p} (\mu - 2j)} \right] \textrm{ mod }2.
\end{align}
Although, calculating these quantities modulo-2 quickly becomes difficult due to the factor of $2^\mu$ in front.
Meanwhile, it can be checked that over the field $\mathbb{F}_2$ the matrix $M$ has a one dimensional null space spanned by the vector $\mathbf{v}_\mathrm{null} = (1,1,\ldots,1)^T$.  It can further be check that there is no such vector $\mathbf{b}\in\mathbb{F}_2^p$ such that $M\mathbf{b}=\mathbf{v}_\mathrm{null}$.  Hence, for each natural number $\mu$ the vector $\mathbf{a}_\mu$ will always have at least one nonzero entry.  Therefore, for each positive integer $n$ and prime number $p$, at least one of the functions $\mathsf{Mod}_{p,j}:\mathbb{F}_2^n\rightarrow\mathbb{F}_2$ has degree-$n$ ANF. 
\end{proof}

\subsection{Functions with maximal periodic Fourier sparsity: A proof of Lemma~\ref{lem:pfslem_label}}
\label{sec:proof_of_ntic_lemma}

\pfslem*
\begin{proof}
By Eq.~(\ref{eq:periodic_fourier}), the periodic Fourier decomposition of a general function $f$ must satisfy
\begin{align}
\exp\left({i\pi\sum_{\mathbf{p}\in\mathbb{F}_2^n} (\mathbf{p}\cdot\mathbf{x}) \phi_\mathbf{p}}\right) = (-1)^{f(\mathbf{x})+f(\mathbf{0})}.
\end{align}
Notice that we have factored out a $\pi$ from each of the $2^n$ real angles $\phi_\mathbf{p}$.  If any linear function $\mathbf{p}\cdot\mathbf{x}$ is unnecessary for the quantum protocol, a solution will yield $\phi_\mathbf{p}\in\mathbb{Z}$.  The corresponding measurements would either have no effect or can accounted for via mod-2 linear post-processing. 

By the Mobius inversion formula, the angles $\phi_\mathbf{p}$ can be related to the coefficients $a_\mathbf{y}$ of the ANF.  We do this by first observing that 
\begin{align}
\prod_{\mathbf{x}\subseteq\mathbf{y}} e^{i\pi\sum_{\mathbf{p}\in\mathbb{F}_2^n} (\mathbf{p}\cdot\mathbf{x}) \phi_\mathbf{p}} &= \prod_{\mathbf{x}\subseteq\mathbf{y}}(-1)^{f(\mathbf{x})} \\
\Rightarrow  \exp\left\{i\pi\sum_{\mathbf{p}\in\mathbb{F}_2^n}\left(\sum_{\mathbf{x}\subseteq\mathbf{y}} (\mathbf{p}\cdot\mathbf{x})\right) \phi_\mathbf{p}\right\} &= (-1)^{a_\mathbf{y}}.
\end{align}
  Note that the term in the sum on the left hand side corresponding to $\mathbf{p}=\mathbf{0}$ contributes nothing to the sum.  Therefore, the angles $\phi_\mathbf{p}$ are solutions to the system of linear equations
\begin{align}
\sum_{\substack{\mathbf{p}\in\mathbb{F}_2^n\\ \mathbf{p}\neq \mathbf{0}}}M_{\mathbf{y},\mathbf{p}}\phi_\mathbf{p} = a_\mathbf{y} + 2\mathbb{Z},
\end{align}
where $M_{\mathbf{y},\mathbf{p}}=\sum_{\mathbf{x}\subseteq\mathbf{y}} (\mathbf{p}\cdot\mathbf{x})$.  Furthermore, the $2\mathbb{Z}$ on the right hand side denotes that we may add any even integer to the $a_\mathbf{y}$ and still obtain a valid periodic Fourier decomposition of the function.  We thus wish to determine the maximum number of angles $\phi_\mathbf{p}\in\mathbb{Z}$ that we can obtain from some choice of inhomogenous term in the above equations.

To solve this linear system, we first derive a general expression for the matrix $M$ and then determine an expression for it's inverse.  First, notice that
\begin{align}
\sum_{\mathbf{x}\subseteq\mathbf{y}} (\mathbf{p}\cdot\mathbf{x}) =
\begin{cases}
0 \textrm{ if }\mathbf{p}\cdot_\mathbb{R}\mathbf{y}=0 \\
2^{|\mathbf{y}|-1} \textrm{ if }\mathbf{p}\cdot_\mathbb{R}\mathbf{y}\neq 0
\end{cases},
\end{align}
where the $\cdot_\mathbb{R}$ denotes that the arithmetic performed in the dot product is over the reals (i.e., $\mathbf{p}\cdot_\mathbb{R}\mathbf{y} = p_1y_1 + \cdots + p_ny_n$).  As a slight abuse of notation, let $\mathbf{p}\cap\mathbf{y} = \textrm{supp}(\mathbf{p})\cap\textrm{supp}(\mathbf{y})$ where $\textrm{supp}(\mathbf{p}) = \{j\in[n]~|~p_j=1\}$.  Then the matrix $M$ can be expressed as
\begin{align}
M_{\mathbf{y},\mathbf{p}} = 2^{|\mathbf{y}|-1}\chi_{\mathbf{p}\cap\mathbf{y}},
\end{align}
where $\chi_{\mathbf{s}\cap\mathbf{y}}=0$ if $\mathbf{s}\cap\mathbf{y} = \varnothing$ and $\chi_{\mathbf{s}\cap\mathbf{y}}=1$ if $\mathbf{s}\cap\mathbf{y} \neq \varnothing$.  The matrix $\chi_{\mathbf{p}\cap\mathbf{y}}$ has the same fractal structure as the Sierpinski gasket, which can be generated by cellular automata rule 195.  Furthermore, this matrix has inverse given by
\begin{align}
\left[M^{-1}\right]_{\mathbf{p},\mathbf{y}} = \frac{(-1)^{\mathbf{p}\cdot\mathbf{y}+1}}{2^{|\mathbf{y}|-1}}(1-\chi_{(\mathbf{p}\oplus\mathbf{1})\cap(\mathbf{y}\oplus\mathbf{1})}).
\label{eq:sierpinski_inverse}
\end{align}
We explicitly show that this is the inverse in Appendix~\ref{sec:proof_sierpinski_inverse}.

Finally, we can use this inverse to find a general expression for the angles $\phi_\mathbf{p}$.  We have that
\begin{align}
\phi_\mathbf{p} = \sum_{\substack{\mathbf{y}\in\{0,1\}^n\\ \mathbf{y}\neq\mathbf{0}}} \frac{(-1)^{\mathbf{p}\cdot\mathbf{y}+1}}{2^{|\mathbf{y}|-1}}(1-\chi_{(\mathbf{p}\oplus\mathbf{1})\cap(\mathbf{y}\oplus\mathbf{1})}) a_\mathbf{y}.
\end{align}
It then follows that if $a_\mathbf{1}=1$, then $\forall\mathbf{p}\in\{0,1\}^n\backslash\mathbf{0}$, $\phi_\mathbf{p}\not\in\mathbb{Z}$.

We prove this by contradiction.  Suppose that $a_\mathbf{1}=1$ and that there is some $\mathbf{p}\in\mathbb{F}_2^n\backslash\mathbf{0}$ such that $\phi_\mathbf{p}\in\mathbb{Z}$.  Notice that for every $\mathbf{p}\in\{0,1\}^n\backslash\mathbf{0}$, $\chi_{\mathbf{p}\cap\mathbf{1}}=1$ and thus $\chi_{(\mathbf{p}\oplus\mathbf{1})\cap\mathbf{0}}=0$.  Therefore, for every $\mathbf{p}\in\mathbb{F}_2^n\backslash\mathbf{0}$ we have that
\begin{align}
 2^{n-1}\phi_\mathbf{p}&= 2^{n-1} \sum_{\substack{\mathbf{y}\in\{0,1\}^n\\ \mathbf{y}\neq\mathbf{0},\mathbf{1}}} [M^{-1}]_{\mathbf{p},\mathbf{y}} a_\mathbf{y} + (-1)^{|\mathbf{p}|+1} \\
&= 2m\pm 1
\end{align}
for some $m\in\mathbb{Z}$.  Now if $\phi_\mathbf{p}\in\mathbb{Z}$, then $2^{n-1}\phi_\mathbf{s}\in 2\mathbb{Z}$. However, the right hand side of the equation is an odd number.  Contradiction.  Therefore, $\forall\mathbf{p}\in\{0,1\}^n\backslash\mathbf{0}$, $\phi_\mathbf{p}\not\in\mathbb{Z}$.  Hence, any Boolean function containing the monomial $x_1\cdots x_n$ in its ANF has periodic Fourier sparsity $\hat{p}_f=2^n-1$.
\end{proof}

\subsection{Proof of Eq.~\eqref{eq:sierpinski_inverse} \label{sec:proof_sierpinski_inverse}}

\begin{lemma}
Let $M_{\mathbf{y},\mathbf{s}} = 2^{|\mathbf{y}|-1}\chi_{\mathbf{s}\cap\mathbf{y}},$, then $\left[M^{-1}\right]_{\mathbf{s},\mathbf{y}} = \frac{(-1)^{\mathbf{s}\cdot\mathbf{y}+1}}{2^{|\mathbf{y}|-1}}(1-\chi_{(\mathbf{s}\oplus\mathbf{1})\cap(\mathbf{y}\oplus\mathbf{1})})$.
\end{lemma}

\begin{proof}
Here we will make use of a one-to-one mapping between elements of $\mathbb{F}_2^n$ and the power set of $[n]$, denoted $\mathcal{P}([n])$.  Namely, each binary string $\mathbf{s}\in\mathbb{F}_2^n$ is uniquely specified by it's support $\mathrm{supp}(\mathbf{s}) = \{j\in[n]\mid s_j=1\}$, which is a subset of $[n]$.  For the binary strings $\mathbf{s}\in\mathbb{F}_2^n$ and $\mathbf{y}\in\mathbb{F}_2^n$, let us denote their corresponding supports as $S\subseteq[n]$ and $Y\subseteq[n]$, respectively.  Notice that $\mathrm{supp}(\mathbf{s}\oplus \mathbf{1}) = \bar{S}$ where $\bar{S}=[n]\backslash S$ denotes the compliment of $S$ in $[n]$.  Furthermore, $(-1)^{\mathbf{s}\cdot\mathbf{y}} = (-1)^{|S\cap Y|}$.  In this notation, Eq.~(\ref{eq:sierpinski_inverse}) can be written
\begin{align}
[M^{-1}]_{S,Y} = \frac{(-1)^{|S\cap Y| + 1}}{2^{|Y|-1}}(1-\chi_{\bar{S}\cap\bar{Y}}).
\end{align}

Notice that $1-\chi_{\bar{S}\cap\bar{Y}} = 1$ if and only if $\bar{S}\cap\bar{Y} = \varnothing$.  By de-Morgan's law this is equivalent to $S\cup Y = [n]$.  Therefore, we may write
\begin{align}
[M^{-1}M]_{S,S'} &= \sum_{Y\subseteq[n]} {(-1)^{|S\cap Y| + 1}}(1-\chi_{\bar{S}\cap\bar{Y}})\chi_{Y\cap S'} \\
&= -\sum_{\substack{Y\subseteq[n] \\ Y\cup S = [n] \\ Y\cap S' \neq \varnothing}} {(-1)^{|S\cap Y| }}.
\end{align}
Since $Y\cup S = [n]$, every term can be written $Y = \bar{S}\cup A$ for some $A\in \mathcal{P}(S)$ satisfying $(\bar{S}\cup A)\cap S' \neq \varnothing$.  Denoting the set of all such subsets $A$ as $K(S,S')$ we may rewrite the above expression as
\begin{align}
    [M^{-1}M]_{S,S'} = -\sum_{A\in K(S,S')} (-1)^{|A|} 
\end{align}

We now analyze what this sum evaluates to based on how $S'$ relates to $S$.
\begin{itemize}
\item[(1)] If $S'\not\subseteq S$, then $K(S,S') = \mathcal{P}(S)$.  The summation gives,
\begin{align}
[M^{-1}M]_{S,S'} &= - \sum_{A\in \mathcal{P}(S)} (-1)^{|A|} \\
&= - \sum_{k=0}^{|S|} {|S|\choose k}(-1)^{k} \\
&= -(1-1)^{|S|} = 0.
\end{align}
\item[(2)] If $S'\subsetneq S$, then we get $K(S,S') = \{A\in \mathcal{P}(S)~|~A\cap S' \neq \varnothing\}$.  Note that the number of subsets $A\in K(S,S')$ with $|A|=k$ is $\sum_{j=1}^k {|S'|\choose j}{|S|-|S'| \choose k-j} = {|S|\choose k}-{|S|-|S'|\choose k}$, by the Vandermonde identity.  Thus we have,
\begin{align}
[M^{-1}M]_{S,S'} &= - \sum_{A\in K(S,S')} (-1)^{|A|} \\
&= -\sum_{k=1}^{|S|} \left( {|S|\choose k}-{|S|-|S'|\choose k} \right)(-1)^k \\
&= -\left[(1-1)^{|S|}-1 -\left((1-1)^{|S|-|S'|}-1 \right) \right] \\
&= 0.
\end{align}
\item[(3)] If $S'=S$, then $K(S,S') = \mathcal{P}(S)\backslash\varnothing$.  Thus we have,
\begin{align}
[M^{-1}M]_{S,S'} &= - \sum_{A\in K(S,S')} (-1)^{|A|} \\
&= -\sum_{k=1}^{|S|}{|S| \choose k}(-1)^{k} \\
&= -\left((1-1)^{|S|} - 1\right) = 1.
\end{align}
\end{itemize}
Therefore,
\begin{align}
[M^{-1}M]_{S,S'} = \delta_{S,S'}.
\end{align}
\end{proof}

\section{MBQC realizations of quantum circuits}
\label{sec:MBQC-1QC_correspondence}

The state $|\mathrm{GHZ}_N\rangle$ can be expressed as the following matrix product state
\begin{align}
    |\mathrm{GHZ}_N\rangle = \sum_{j_1,\ldots,j_N=0}^1 \langle 0 |\Pi_X^{(j_N)}\cdots \Pi_X^{(j_1)} |0 \rangle~~|j_1\cdots j_N\rangle.
\end{align}
Here, $\Pi_X^{(j)} = H|j\rangle\langle j|H$. Since 
\begin{align}
|m(\theta)\rangle = \frac{e^{-i\theta/2}|0\rangle + (-1)^m e^{i\theta/2}|1\rangle}{\sqrt{2}}
\end{align}
it follows that
\begin{subequations}
\begin{align}
    \langle\mathbf{m}(\boldsymbol{\theta})|\mathrm{GHZ}_N\rangle &=  \frac{1}{\sqrt{2^N}} \langle 0 |\prod_{k=1}^N \left(e^{-i\theta_k/2}\Pi_X^{(0)} + (-1)^{m_k} e^{i\theta_k/2}\Pi_X^{(1)}\right) |0 \rangle\\
    &=  \frac{1}{\sqrt{2^N}} \langle 0 |\prod_{k=1}^N X^{m_k}R_X(\theta_k) |0 \rangle \\
    &=\frac{1}{\sqrt{2^N}} \langle \bigoplus_{k=1}^N m_k | \prod_{k=1}^N R_X(\theta_k) |0\rangle.
\end{align}
\end{subequations}
Therefore,
\begin{align}
    |\langle \mathbf{m}(\boldsymbol{\theta})|\mathrm{GHZ}_N\rangle|^2 &\propto |\langle \bigoplus_{j=1}^N m_j | \prod_{j=1}^N R_X( \theta_j) |0\rangle |^2.
\end{align}

The state $|\mathrm{1DC}_{2N+1}\rangle$ can be expressed as the following matrix product state
\begin{align}
    |\mathrm{1DC}_{2N+1}\rangle = \sum_{j_1,\ldots,j_{2N+1}=0}^1 \langle 0 | \Pi_X^{(j_{2N+1})}\prod_{k=1}^{N} \Pi_Z^{(j_{2k})} \Pi_X^{(j_{2k-1})} |0\rangle~~|j_1\cdots j_{2N+1}\rangle.
\end{align}
Following a similar calculation sketched above, we get
\begin{align}
    \langle \mathbf{m}(\boldsymbol{\theta}) | \mathrm{1DC}_{2N+1} \rangle = \frac{1}{\sqrt{2^N}} \langle 0 | X^{m_{2N+1}} R_X(\theta_{2N+1}) \prod_{k=1}^N Z^{m_{2k}} R_Z(\theta_{2k}) X^{m_{2k-1}} R_X(\theta_{2k-1}) |0\rangle. 
\end{align}
Pushing all Pauli operators to the left gives
\begin{align}
    \langle \mathbf{m}(\boldsymbol{\theta}) | \mathrm{1DC}_{2N+1} \rangle = \frac{1}{\sqrt{2^N}} \langle 0 | \left(X^{m_{2N+1}}\prod_{l=1}^N Z^{m_{2k}}X^{m_{2k-1}}\right) \left(R_X(\tilde{\theta}_{2N+1}) \prod_{k=1}^N R_Z(\tilde{\theta}_{2k}) R_X(\tilde{\theta}_{2k-1})\right) |0\rangle
\end{align}
where $\tilde{\theta}_j = (-1)^{\sum_{k=1}^n A_{jk}m_k}\theta_{j}$ and
\begin{align}
    A_{j,k} = 
    \begin{cases}
    1 &\textrm{if }k<j\textrm{ and }k\neq j\textrm{ mod }2 \\
    0 & \textrm{otherwise}
    \end{cases}.
\end{align}
Absorbing all the Pauli operators into the bra $\langle 0 |$ simply flips this state and adds phases conditional on $\mathbf{m}$.  Therefore, we obtain
\begin{align}
    |\langle \mathbf{m}(\boldsymbol{\theta})| \mathrm{1DC}_{2N+1}\rangle|^2 \propto 
    |\langle \bigoplus_{j=0}^N m_{2j+1} | R_X(\tilde{\theta}_{2N+1}) \prod_{j=1}^N R_Z(\tilde{\theta}_{2j})R_X(\tilde{\theta}_{2j-1}) |0\rangle|^2.
\end{align}

\section{Adaptive $l2$-MBQC is universal}
\label{sec:adaptive_universality}

In this section we show how to implement any commutative $l2$-1QC via adaptive $l2$-MBQC in constant-time.  It then follows that adaptive $l2$-MBQC can compute any Boolean function, albeit not necessarily efficiently. 

\begin{theorem}
Any nonadaptive $l2$-MBQC that can be performed with an $N$-qubit GHZ state can be performed adaptively with a $(2N+1)$-qubit 1D cluster state and overall resource costs $L_C=N$, $L_Q=2N+1$, $T_C=2$, and $T_Q=3$.
\label{thm:aMQC_universal}
\end{theorem}

\begin{proof}
Universality of the 1D cluster state follows from the well known fact that by measuring every other qubit in the Pauli-$X$ basis recovers a GHZ state up to Pauli-$X$ corrections on the unmeasured qubits \cite{briegel2001persistent}.  However, we will instead describe explicitly how to implement any commutative $l2$-1QC via adaptive $l2$-MBQC as some ideas will be of use to us later.

A Boolean function with periodic Fourier decomposition $\cos(\sum_{j=1}^N (\mathbf{p}_j\cdot\mathbf{x})\phi_j)$ can be computed by the $l2$-1QC $U(\mathbf{x})=\prod_{j=1}^N R_X((\mathbf{p}_j\cdot\mathbf{x})\phi_j)$. Since any binary number $x\in\mathbb{F}_2$ satisfies $x = (1-(-1)^{x})/2$, we can rewrite $U(\mathbf{x})$ as
\begin{align}
    U(\mathbf{x}) = R_X(\sum_{j=1}^N \phi_j/2)\prod_{j=1}^N R_X((-1)^{(\mathbf{p}_j\cdot\mathbf{x})+1}\phi/2).
\end{align}

Now consider a $(2N+1)$-qubit 1D cluster state.  The above circuit can be implemented in a measurement-based manner via the following sequence of measurements,
\begin{itemize}
    \item[(1)] Each qubit with label $j=0\textrm{ mod 2}$ is measured in the Pauli-$X$ basis (i.e., $\theta_j=0$) and each measurement outcome $m_j\in\mathbb{F}_2$ is returned to the mod-2 linear classical side-processor.
    \item[(2)] For each qubit with label $j=1\textrm{ mod }2$---except $j=2N+1$---the side-processor computes and returns the bit
    \begin{align}
        s_j = \left(\bigoplus_{\substack{k<j \\ k=0\textrm{ mod 2}}} m_k\right)\oplus (\mathbf{p}_{(j-1)/2}\cdot\mathbf{x}).
    \end{align}
    Each qubit is then measured in the eigenbasis of $X(\theta_j)$ where $\theta_j = \frac{1}{2}(-1)^{s_j+1}\phi_\mu$.  Meanwhile, for qubit $2N+1$ the side-processor computes and returns the bit
    \begin{align}
         s_{2N+1} = \bigoplus_{\substack{k<{2N+1} \\ k=0\textrm{ mod 2}}} m_k.
    \end{align}
    The qubit then is measured in the eigenbasis of $X(\theta_{2N+1})$ where $\theta_{2N+1} = (-1)^{s_{2N+1}}\sum_{j=1}^N \phi_j/2$
    The corresponding measurement outcomes $m_j\in\mathbb{F}_2$ are returned to the side-processor, which then computes and returns the computational output
    \begin{align}
        y = \left(\bigoplus_{j=1\textrm{ mod }2} m_j\right)\oplus f(\mathbf{0}).
    \end{align}
\end{itemize}
It follows from Eq.~\ref{eq:adaptive_MBQC_su2} that $y=f(\mathbf{x})$.
\end{proof}

This demonstrates that adaptive $l2$-MBQC with 1D cluster states can compute any function $f:\mathbb{F}_2^n\rightarrow\mathbb{F}_2$ using $L_Q=2\hat{p}_f+1$ many qubits.

\section{Quantum signal processing}
\label{sec:Previous_1QC_Work}

The quantum signal processing technique \cite{low2016methodology, low2017optimal, haah2019product} gives necessary and sufficient conditions for a single qubit unitary $U(\phi)$ to be implemented by a sequence of $L$ many rotations by an angle $\phi$ about some axes in the $XY$-plane of the Bloch sphere.  Namely, a single qubit unitary with Pauli decomposition $U(\phi) = A(\phi)I + iB(\phi)X + iC(\phi)Y + i D(\phi) Z$ can be decomposed as the sequence of rotations
\begin{align}
    U(\phi) = \left(\prod_{j=1}^L R_Z(\xi_j)R_X(\phi)R_Z(\xi_j)^\dagger\right)R_Z(\xi_0).
    \label{eq:QSP_Decomp}
\end{align}
for some angles each $\xi_j\in[0,2\pi)$ whenever---under the change of variables $z=e^{i\phi/2}$---the functions $A(z),B(z),C(z),D(z)$ have the following properties.
\begin{itemize}
    \item[(QSP1)] $A(z)^2 + B(z)^2 + C(z)^2 + D(z)^2=1$.
    \item[(QSP2)] $A,B,C,D$ are Laurent polynomials in $z$ of degree at most $L$ and at least one has degree equal to $L$.
    \item[(QSP3)] $A,B,C,D$ are even/odd functions of $z$ whenever $L$ is even/odd, respectively.
    \item[(QSP4)] $A$ and $D$ are reciprocal, i.e. $A(z)=A(1/z)$, whereas $B$ and $C$ are antireciprocal, i.e. $B(z)=-B(1/z)$.
\end{itemize}
Furthermore, there is an efficient algorithm to determine the decomposition in Eq.~(\ref{eq:QSP_Decomp}) discussed in Ref.~\cite{haah2019product}.  

Properties (QSP2), (QSP3), and (QSP4) imply that $A(\phi)$ and $B(\phi)$  resemble Fourier cosine and sine series, respectively. Namely, if $L$ is odd, then
\begin{subequations}
\begin{align}
    A(\phi) &= \sum_{j=1}^L a_j \cos(j\phi/2) \\
    B(\phi) &= \sum_{j=1}^L b_j\sin(j\phi/2)
\end{align}
\end{subequations}
where $a_j=b_j=0$ for all even values of $j$.

Ref.~\cite{maslov2020quantum} used quantum signal processing to construct 1QCs that compute symmetric Boolean functions.  Their general construction is based on setting $\phi_{|\mathbf{x}|} = \pi |\mathbf{x}|/ (n+1)$ and determining the coefficients ${a_j,b_j}_{j=1}^{L}$ by solving the system of linear equations 
\begin{subequations}
\label{eq:QSP_Linear_System}
\begin{align}
    A(\phi_{|\mathbf{x}|}) &= 1-f(|\mathbf{x}|) \\
    B(\phi_{|\mathbf{x}|}) &= f(|\mathbf{x}|) \\
    \frac{dA}{d\phi}\Bigr|_{\phi_{|\mathbf{x}|}} &= 0 \\
    \frac{dB}{d\phi}\Bigr|_{\phi_{|\mathbf{x}|}} &= 0
\end{align}
\end{subequations}
for each value of the Hamming weight $|\mathbf{x}|\in\{0,\ldots,n\}$.  The latter two equations guarantee the existence of polynomials $C(z)$ and $D(z)$ such that condition (QSP1) is satisfied.  We will discuss in detail how $C(z)$ and $D(z)$ are determined in the next subsection.

While Ref.~\cite{maslov2020quantum} showed that setting $\phi_{|\mathbf{x}|} = \pi |\mathbf{x}|/ (n+1)$ is sufficient to construct a 1QC with $L=4n+1$, it is not always optimal.  For functions with a high degree of symmetry, the angles $\phi_{|\mathbf{x}|}$ can be chosen so as to reduce the size of the above system of equations.

\subsection{The mod-$p$ functions}

The mod-$p$ functions have a high of symmetry, which can be leveraged in the quantum signal processing to construct 1QCs with $L=2p-1$.  This would imply that the overall depth of the 1QC is $O(n)$ and that the circuit consists of a constant number of blocks of mutually commuting unitaries occurring sequentially.  Such a construction implies the constant-time adaptive $l2$-MBQC protocol described in Thm.~\ref{thm:mod_p_mbqc}.

The symmetry we leverage is the invariance of $\mathsf{Mod}_{p,j}(|\mathbf{x}|)$ to adding multiples of $p$ to the Hamming weight.  Hence, it is reasonable to make the angle $\phi_{|\mathbf{x}|}/2$ to also have this property.  Analyzing the situation when $\phi_{|\mathbf{x}|} = 4\pi|\mathbf{x}|/p$ gives the following result.
\begin{lemma}
A single qubit computation of the mod-$p$ function can be achieved using the quantum signal processing scheme with $L=2p-1$ and $\phi_{|\mathbf{x}|} = 4\pi|\mathbf{x}|/p$.
\end{lemma}
\begin{proof}
Since both the truth table of the mod-$p$ function and the functions $\cos(j\phi_{|\mathbf{x}|}/2)$ and $\sin(j\phi_{|\mathbf{x}|}/2)$ have period $p$ with respect to $|\mathbf{x}|$, Eqs.~\ref{eq:QSP_Linear_System} reduce to the following $2p$ equations in $2p$ variables.
\begin{subequations}
\begin{align}
    \sum_{j=1}^p a_{2j-1} \cos\left((2j-1)\frac{2\pi w}{p}\right)&= \delta_{w,0}\label{eq:mod_p_qsp_linsyst_1}\\
    \sum_{j=1}^p b_{2j-1} \sin\left((2j-1)\frac{2\pi w'}{p}\right) &= 0 \label{eq:mod_p_qsp_linsyst_2}\\
    \sum_{j=1}^p a_{2j-1} (2j-1) \sin\left((2j-1)\frac{2\pi w'}{p}\right) &=0 \label{eq:mod_p_qsp_linsyst_3}\\
    \sum_{j=1}^p b_{2j-1} (2j-1) \cos\left((2j-1)\frac{2\pi w}{p}\right) &= 0 \label{eq:mod_p_qsp_linsyst_4}
\end{align}
\end{subequations}
where above $w$ runs from $0$ to $(p-1)/2$ in integer steps and $w'$ runs from $1$ to $(p-1)/2$ in integer steps.  This linear system has a unique solution, which guarantees that $L=2p-1$ for any size of input $n$.  The method of \cite{haah2019product} can then be used to determine the functions $C(\phi)$, $D(\phi)$, and moreover the angles $\{\xi_j\}_{j=0}^{2p-1}$ to a desired numerical precision. 
\end{proof}

\section{The OR-reduction and ancillary qubits for counting}
\label{sec:previous_work_as_cluster_mbqc}

\subsection{The OR-reduction for computing general symmetric Boolean functions}

Reducing the problem of computing $\mathsf{OR}_n:\mathbb{F}_2^n\rightarrow\mathbb{F}_2$ to the task of computing $\mathsf{OR}_{\lceil\log_2(n)\rceil}:\mathbb{F}_2^{\lceil\log_2(n)\rceil}\rightarrow\mathbb{F}_2$ follows from the observation that the $\mathsf{OR}$ of the $n$ bits in the string $\mathbf{x}$ equals the $\mathsf{OR}$ of the $\lceil \log_2(n) \rceil$ bits in the binary string representing the number $|\mathbf{x}|$.  In Ref.~\cite{takahashi2016collapse} a constant-depth quantum circuit with quantum fan-out gates implementing this task was introduced.  Here we describe this scheme as an adaptive $l2$-MBQC using a 1D cluster state.  The algorithm leverages an old result of Ref.~\cite{hoyer2005quantum} which uses an ancillary space of $\kappa = \lceil\log_2(n)\rceil$ many qubits to count---in binary---the Hamming weight of the input string $\mathbf{x}\in\mathbb{F}_2^n$.

 The original construction using a $\kappa$-qubit ancillary space works as follows.  Conditioned on each bit $x_j$ for each $j\in[n]$ of the string $\mathbf{x}$, apply to the ancillary space the unitary
\begin{align}
M= \sum_{\mathbf{z}\in\mathbb{F}_2^\kappa} |\mathrm{bi}(\mathrm{de}(\mathbf{a})+1)\textrm{ mod }2^\kappa\rangle\langle \mathbf{a}|.
\end{align}
Here, ``$\mathrm{de}(\cdot)$" is a function that maps a binary number to its decimal representation and ``$\mathrm{bi}(\cdot)$" maps a decimal number to its binary representation. 
Let us call the state of the classical memory register $|\mathbf{x}\rangle$ and the final state of the ancillary space $|\mathbf{a}_\mathrm{f}\rangle$. Hence, if $\mathbf{x}=\mathbf{0}$ then $\mathbf{a}_\mathrm{f} = \mathbf{0}$ and if $\mathbf{x}\neq\mathbf{0}$, then $\mathbf{a}_\mathrm{f}\neq\mathbf{0}$. So, $\mathsf{OR}(\mathbf{x}) = \mathsf{OR}(\mathbf{a})$.

The while the above circuit is rather complicated, requiring quantum fourier transforms, a similar circuit that implements an equivalent reduction can be implemented in constant depth.  First note that
\begin{align}
    M = U_\mathrm{QFT} D U_\mathrm{QFT}^\dagger
\end{align}
where
\begin{align}
    D = \sum_{\mathbf{a}\in\mathbb{F}_2^n} \exp\left(i\frac{2\pi}{2^\kappa} \sum_{j=0}^{\kappa-1}2^ja_j\right)|\mathbf{a}\rangle\langle\mathbf{a}|.
\end{align}
The unitary $U_{\textrm{OR-reduc}}=\sum_{\mathbf{x}\in\mathbb{F}_2^n}\left( |\mathbf{x}\rangle\langle\mathbf{x}|\otimes M^{|\mathbf{x}|}\right)$ can then be written as
\begin{align}
    U_{\textrm{OR-reduc}} = \sum_{\mathbf{x}\in\mathbb{F}_2^n} |\mathbf{x}\rangle\langle\mathbf{x}|\otimes \left(U_{\mathrm{QFT}} D^{|\mathbf{x}|} U_{\mathrm{QFT}}^\dagger\right).
\end{align}
Notice that
\begin{align}
    D^{|\mathbf{x}|} = \bigotimes_{j=1}^\kappa \sum_{a_j\in\mathbb{F}_2} \exp\left(i\frac{2\pi}{2^\kappa} 2^{j-1}a_j |\mathbf{x}|\right) |a_j\rangle\langle a_j| \propto \bigotimes_{j=1}^\kappa R_Z\left(\frac{2\pi}{2^{j}} |\mathbf{x}|\right)
\end{align}
is a tensor product of single-qubit $Z$-rotations up to a global phase.  Furthermore, since the ancillary qubits are initialized in the $|\mathbf{a}=\mathbf{0}\rangle$ state, we may replace $U_{\mathrm{QFT}}$ with $H^{\otimes n}$.  Furthermore, since in the measurement stage all we care about is if the final string gives all zeros as measurement outcomes, we may replace the final $U_\mathrm{QFT}^\dagger$ with $H^{\otimes n}$.  Hence, the following circuit implements the OR-reduction
\begin{align}
    U_{\mathsf{OR}\textrm{-}\mathrm{reduc}}=\bigotimes_{j=1}^\kappa R_X\left(\frac{2\pi}{2^{j}} |\mathbf{x}|\right).
    \label{eq:or_reduc_1QC}
\end{align}
This circuit is just $\kappa$ many 1QCs of length $T_Q = n$.  Hence, it can be implemented by nonadaptive $l2$-MBQC on $\kappa$ many uncoupled $n$-qubit GHZ states.  

Once the reduction has been implemented by a single round of nonadaptive $l2$-MBQC, the function $\mathsf{OR}_{\lceil\log_2(n)\rceil}:\mathbb{F}_2^{\lceil\log_2(n)\rceil}\rightarrow\mathbb{F}_2$ can be computed via another round of $l2$-MBQC on a $2^{\lceil \log_2(n) \rceil}-1 = O(n)$ qubit GHZ state where the inputs are determined by the previous measurement outcomes.  Namely, given the string $\mathbf{a}$ we consider the $l2$-MBQC implementation of the $l2$-1QC given by the single-qubit circuit
\begin{align}
\label{eq:OR_anc_1QC}
(iX)^{\mathsf{OR}_\kappa(\mathbf{a})} = \prod_{\substack{S\subseteq[\kappa] \\ S\neq\varnothing}}R_X\left(  \pi (-1)^{|S|-1} \frac{2^{\kappa-|S|+1}-1}{2^{\kappa-1}} \left( \bigoplus_{j\in S} a_j \right) \right).
\end{align}

This whole scheme can be expressed as an adaptive $l2$-MBQC with a 1D cluster state as follows.
\begin{theorem}
The function $\mathsf{OR}:\mathbb{F}_2^n\rightarrow\mathbb{F}_2$ can be computed via adaptive $l2$-MBQC on a 1D cluster state with resource costs $L_Q = 2\lceil \log_2(n)\rceil(n+1) + 2^{\lceil \log_2(n)\rceil +1} - 1$, $T_Q=3$, $L_C=(n+2)\lceil \log_2(n)\rceil + 2^{\lceil \log_2(n)\rceil}$, and $T_C=3$.
\end{theorem}
\begin{proof}
Consider a ($2\kappa(n+1) + 2^{\kappa +1} - 1$)-qubit 1D cluster state, where $\kappa = \lceil \log_2(n)\rceil$.  The circuit in Eq.~(\ref{eq:or_reduc_1QC}) followed by the circuit in Eq.~(\ref{eq:OR_anc_1QC}) can be implemented in a measurement-based manner via the following sequence of measurements.
\begin{itemize}
    \item[(1)] Each qubit with label $q=2(\mu-1)(n+1) + 2j$ for each $j\in[n]$ and $\mu\in[\kappa]$ is measured in the Pauli $X$-basis.  Meanwhile, each qubit with label $q=2\mu(n+1)$ for $\mu\in[\kappa]$ is measured in the Pauli $Z$-basis.  Furthermore, each qubit with label $q= 2\kappa(n+1) + 2j$ for each $j\in[2^\kappa-1]$ is measured in the Pauli-$X$ basis.  Each measurement outcome $m_q\in\mathbb{F}_2$ is returned to the mod-2 linear classical side-processor.
    \item[(2)] For each qubit with label $q=2(\mu-1)(n+1) + 2j -1$ for each $j\in[n]$ and $\mu\in[\kappa]$, the side-processor computes and returns the bit
    \begin{align}
        s_q = \left( \bigoplus_{l=0}^{j-1} m_{2l + 2(\mu-1)(n+1)} \right)\oplus x_j.
    \end{align}
    Here we take the convention that $m_0=0$.  Each qubit is then measured in the eigenbasis of $X(\theta_q)$ where $\theta_q = (-1)^{s_q} \pi/2^{\mu+1}$.  Meanwhile, for each qubit with label $q=2(\mu-1)(n+1) + 2n+1$ for each $\mu\in[\kappa]$ the side-processor computes and returns the bit
    \begin{align}
        s_q = \bigoplus_{l=0}^{n} m_{2l + 2(\mu-1)(n+1)}.
    \end{align}
    Each qubit is then measured in the eigenbasis of $X(\theta_q)$ where $\theta_q = (-1)^{s_q+1} n\pi/2^{\mu}$.
    \item[(3)] For each qubit with label $q = 2\kappa (n+1) + 2j-1$ for each $j\in[2^\kappa - 1]$ the side-processor computes and returns the bit
    \begin{align}
        s_q = \left(\bigoplus_{0\leq l<j} m_{2\kappa(n+1)+ 2l} \right) \oplus \left( \bigoplus_{l\in\mathrm{Set}(j)} \left[ 
        \bigoplus_{{a}\in[n]} m_{2(l-1)(n+1)+2a-1}\right]\oplus m_{2(l-1)(n+1)}\oplus m_{2l(n+1)}\right).
    \end{align}
    Each qubit is then measured in the eigenbasis of $X(\theta_q)$ where
    \begin{align}
        \theta_q = \frac{\pi}{2^n} (-1)^{|\mathrm{Set}(j)|} \left( 2^{n- |\mathrm{Set}(j)| + 1} - 1 \right)(-1)^{s_q}.
    \end{align}
    Meanwhile, for the qubit with label $q=2\kappa(n+1) + 2^{\kappa+1}-1$ the side-processor computes and returns the bit
    \begin{align}
        s_q = \bigoplus_{l=0}^{2^{\kappa}-1} m_{2\kappa(n+1) + 2l}.
    \end{align}
    The qubit is then measured in the eigenbasis of $X(\theta_q)$ where
    \begin{align}
        \theta_q = (-1)^{s_q+1} \frac{\pi}{2^n}\sum_{\substack{S\subseteq[n] \\ S\neq \varnothing}} (-1)^{|S|-1} \left( 2^{n-|S|+1}-1\right).
    \end{align}
    Each measurement outcome $m_q\in\mathbb{F}_2$ is returned to the mod-2 linear classical side-processor, which then computes and returns the computational output
    \begin{align}
        y=\left( \bigoplus_{j\in[2^\kappa]} m_{2\kappa(n+1) + 2j-1} \right) \oplus m_{2\kappa(n+1)}.
    \end{align}
\end{itemize}
It follows that $y= \mathsf{OR}_n(\mathbf{x})$.
\end{proof}

\subsection{The OR-reduction for mod-$p$ functions}

A similar reduction was derived in Ref~\cite{moore1999quantum} for the mod-$p$ functions.  Namely, the function $\mathsf{Mod}_{p,0} :\mathbb{F}_2^n \rightarrow \mathbb{F}_2$ can be reduced to a computation of $\mathsf{OR}_\kappa:\mathbb{F}_2^\kappa \rightarrow \mathbb{F}_2$ where $\kappa = \lceil\log_2(p)\rceil$. This reduction uses an ancillary space of $\kappa$ many qubits to count---in binary---the number of 1's in the string $\mathbf{x}\in\mathbb{F}_2^n$, resetting the counter to $\mathbf{0}$ each time it reaches the value $p$. 

For convenience, let us denote the state $|\mathbf{a}\rangle \in (\mathbb{C}^2)^{\otimes \kappa}$ by a single ket with $\mathbf{a}$'s decimal representation, $|a\rangle$.  Conditioned on each bit $x_j$ for each $j\in[n]$ in the string $\mathbf{x}$, apply to the ancillary space the unitary
\begin{align}
    M = 
    \begin{cases}
    |a\rangle \mapsto |(a+1)\textrm{ mod }p\rangle &\textrm{if }a< p\\
    |a\rangle \mapsto |a\rangle &a\geq p
    \end{cases}.
\end{align}
This unitary is diagonal in the discrete Fourier basis.  Namely, let 
\begin{align}
    U_{\mathrm{DFT}}^{(p)} = \left(\frac{1}{\sqrt{p}}\sum_{a,b =0}^{p-1} \exp(-i{2\pi ab}/{p}) |a\rangle \langle b|\right) + \sum_{c=p}^{2^\kappa - 1} |c\rangle\langle c|
\end{align}
 be the discrete Fourier transform on the $p$-dimensional subspace spanned by $\{|a\rangle\}_{a=0}^{p-1}$.  It follows that
\begin{align}
    M = U_{\mathrm{DFT}}^{(p)} D {U_{\mathrm{DFT}}^{(p)}}^\dagger
\end{align}
where
\begin{align}
    D = \sum_{a=0}^{p-1} \exp(i2\pi a/p) |a\rangle \langle a| + \sum_{b=p}^{2^{\lceil \log_2(p)\rceil}-1} |b\rangle\langle b|.
\end{align}
The action of the unitary $M$ on the subspace spanned by $\{|b\rangle\}_{b=p}^{2^{\kappa}-1}$ is irrelevant, which allows us to replace $D$ by the single-qubit rotations
\begin{align}
    D = \bigotimes_{j=1}^{\lceil \log_2(p) \rceil} R_Z\left(\frac{2^j\pi}{p}\right).
\end{align}
Hence the unitary $U_{\mathrm{mod}\textrm{-}p\textrm{-reduc}} =\sum_{\mathbf{x}\in\mathbb{F}_2^n} |\mathbf{x}\rangle\langle\mathbf{x}|\otimes \left(U_{\mathrm{DFT}}^{(p)} D^{|\mathbf{x}|} U_{\mathrm{DFT}}^{(p) \dagger}\right)$ implements the desired reduction.

While for the OR-reduction we were able to replace the quantum Fourier transforms by a simple round of Hadamard gates, a similar simplification is not possible for the mod-$p$-reduction.  This is due to the fact that the eigenstates of $U_{\textrm{DFT}}^{(p)}$ with support on $\{|a\rangle\}_{a=0}^{p-1}$ are entangled.  At best, we can replace the unitary $U_{\mathrm{DFT}}^{(p)}$ by any unitary $V \in \mathrm{U}(2^\kappa)$ satisfying
$V|0\rangle = \frac{1}{\sqrt{p}}\sum_{a=0}^{p-1} |a\rangle$.  Implementing this unitary in the circuit model requires an expected $O(p^2\log^3(p))$ many one and two-qubit gates and hence depth $O(p^2\log^3(p))$.  Similarly, the $l2$-MBQC scheme implementing this reduction will require $T_Q+T_C = O(p^2\log^3(p))$ depending on if we implement the Fourier transforms directly on the quantum hardware to prepare the entangled resource state or in a measurement-based manner using a 2D cluster state.  

We now explicitly give an adaptive $l2$-MBQC protocol that computes $\mathsf{mod}_{3,0}$ using a 2D cluster state resource.  For simplicity, we discuss it's minimal implementation, which requires a $4n+35$-qubit resource state shown in Fig.~\ref{fig:mod_3_2d_mbqc}.
\begin{theorem}
The constant-depth quantum circuit introduced in Ref.~\cite{moore1999quantum} that computes the function $\mathsf{Mod}_{3,0}:\mathbb{F}_2^n\rightarrow\mathbb{F}_2$ can be recast as an adaptive $l2$-MBQC on the 2D resource state depicted in Fig.~\ref{fig:mod_3_2d_mbqc} with resource costs $L_C = n+4$, $L_Q = 4n + 35$,  $T_C = 13$, and $T_Q = 4$.
\end{theorem}
\begin{proof}
Consider the $4n+35$ qubit graph state prepared on the graph shown in Fig.~\ref{fig:mod_3_2d_mbqc}.  The graph is depicted with each vertex positioned on a $2\times (2n+20)$ grid, which we use to label each qubit with a row and column index $(q_{\mathrm{row}},q_{\mathrm{col}})\in[2]\times[2n+20]$.  The circuit $U(\mathbf{x}) = U_{\mathrm{DFT}}^{(p)} D^{|\mathbf{x}|} U_{\mathrm{DFT}}^{(p) \dagger}$, followed by a computation of $\mathsf{OR}_2$, can be implemented in a measurement-based manner via the following sequence of measurements.  
\begin{itemize}
\item[(1)] Each qubit with label $(q_{\mathrm{row}},q_{\mathrm{col}}) =  (j,2k+5)$ for each $j\in[2]$ and $k\in[n+2]$ is measured in the Pauli-$X$ basis.  Moreover, qubits with label  $(q_{\mathrm{row}},q_{\mathrm{col}})=(1,l)$ for $l\in\{4,5,6,2n+10, 2n+ 11, 2n + 12, 2n+17, 2n+ 19\}$ are measured in the Pauli-$X$ basis. Meanwhile, qubits with label $(q_{\mathrm{row}},q_{\mathrm{col}}) = (j,1)$ for $j\in[2]$ are measured in the Pauli-$Y$ basis.  Each measurement outcome $m_{(q_{\mathrm{row}},q_{\mathrm{col}})}\in\mathbb{F}_2$ is returned to the mod-2 linear classical side-processor.
\item[(2)] For the qubits with label $(q_{\mathrm{row}},q_{\mathrm{col}}) = (j,2)$ for $j\in[2]$ the classical side-processor returns the bit
\begin{align}
s_{(j,2)} = m_{(j,1)}.
\end{align}
Each qubit is then measured in the eigenbasis of $X(\theta_{(j,2)})$ where $\theta_{(1,2)} = (-1)^{s_{(1,2)}}\pi/4$ and $\theta_{(2,2)} = (-1)^{s_{(2,2)}}2\theta$.  Here, $\cos(\theta) = 1/\sqrt{3}$ and $\sin(\theta) = \sqrt{2/3}$.  Each measurement outcome $m_{(q_{\mathrm{row}},q_{\mathrm{col}})}\in\mathbb{F}_2$ is then returned to the classical side-processor.
\item[(3)]  For the qubits with label $(q_{\mathrm{row}},q_{\mathrm{col}}) = (j,3)$ for $j\in[2]$ the classical side-processor returns the bit
\begin{align}
s_{(j,3)} = m_{(j,2)}.
\end{align}
Each qubit is then measured in the eigenbasis of $X(\theta_{(j,3)})$ where $\theta_{(j,3)} = (-1)^{s_{(j,3)}+1}\pi/2$.  Each measurement outcome $m_{(q_{\mathrm{row}},q_{\mathrm{col}})}\in\mathbb{F}_2$ is then returned to the classical side-processor.
\item[(4)] For the qubit with label $(2,4)$ the classical side-processor computes and returns the bit
\begin{align}
s_{(2,4)} = m_{(2,1)}\oplus m_{(2,3)}.
\end{align}
The qubit is then measured in the eigenbasis of $X(\theta_{(2,4)})$ where $\theta_{(2,4)} = (-1)^{s_{(2,4)}}\pi/2$.  The corresponding measurement outcome $m_{(2,4)}\in\mathbb{F}_2$ is then returned to the classical side-processor.
\item[(5)] For the qubit with label $(2,5)$ the classical side-processor computes and returns the bit
\begin{align}
s_{(2,5)} = m_{(2,2)}\oplus m_{(2,4)} \oplus m_{(1,1)} \oplus m_{(1,3)}.
\end{align}
The qubit is then measured in the eigenbasis of $X(\theta_{(2,5)})$ where $\theta_{(2,5)} = (-1)^{s_{(2,5)}+1}\pi/4$.  The corresponding measurement outcome $m_{(2,5)}\in\mathbb{F}_2$ is then returned to the classical side-processor.
\item[(6)] For the qubit with label $(2,6)$ the classical side-processor computes and returns the bit
\begin{align}
s_{(2,6)} = m_{(2,1)}\oplus m_{(2,2)} \oplus m_{(2,5)} .
\end{align}
The qubit is then measured in the eigenbasis of $X(\theta_{(2,6)})$ where $\theta_{(2,6)} = (-1)^{s_{(2,6)}+1}\pi/2$.  The corresponding measurement outcome $m_{(2,6)}\in\mathbb{F}_2$ is then returned to the classical side-processor.
\item[(7)] For each qubit with label $(q_{\mathrm{row}},q_{\mathrm{col}}) = (j,2k+6)$ for $j\in[2]$ and $k\in[n]$ the classical side-processor computes and returns the bit
\begin{align}
s_{(j,2k+6)} = \left(\bigoplus_{l=1}^{k+3} m_{(j,2l-1)}\right)\oplus x_k.
\end{align}
Each qubit is then measured in the eigenbasis of $X(\theta_{(j,2k+6)})$ where $\theta_{(j,2k+6)} = (-1)^{s_{(j,2k+6)} + j + 1} \pi/3$.
Meanwhile, for each qubit with label $(q_{\mathrm{row}},q_{\mathrm{col}}) = (j,2n+8)$ for $j\in[2]$ the classical side-processor computes and returns the bit
\begin{align}
s_{(j,2n+8)} = \bigoplus_{l=1}^{n+4} m_{(j,2l-1)}.
\end{align}
Each qubit is then measured in the eigenbasis of $X(\theta_{(j,2n+8)})$ where $\theta_{(j,2n+8)} = (-1)^{s_{(j,2n+8)}+j} n\pi/3$.  Each measurement outcome $m_{(q_{\mathrm{row}},q_{\mathrm{col}})}\in\mathbb{F}_2$ is then returned to the classical side-processor.
\item[(8)] For the qubit with label $(2,2n+10)$ the classical side-processor computes and returns the bit
\begin{align}
s_{(2,2n+10)} = \bigoplus_{l=1}^{n+5}m_{(2,2l-1)}.
\end{align}
The qubit is then measured in the eigenbasis of $X(\theta_{(2,2n+10)})$ where $\theta_{(2,2n+10)} = (-1)^{s_{(2,2n+10)}}\pi/2$.  The corresponding measurement outcome $m_{(2,2n+10)}\in\mathbb{F}_2$ is returned to the classical side-processor.
\item[(9)] For the qubit with label $(2,2n+11)$ the classical side-processor computes and returns the bit
\begin{align}
s_{(2,2n+11)} = \left(\bigoplus_{l=1}^{n+5}m_{(2,2l)}\right) \oplus \left( \bigoplus_{l=4}^{n+5}m_{(1,2l-1)} \right)\oplus m_{(1,1)}\oplus m_{(1,3)}.
\end{align}
The corresponding measurement outcome $m_{(2,2n+11)}\in\mathbb{F}_2$ is returned to the classical side-processor.
\item[(10)] For the qubit with label $(2,2n+12)$ the classical side-processor computes and returns the bit
\begin{align}
s_{(2,2n+12)} = \bigoplus_{l=1}^{n+6}m_{(2,2l-1)}.
\end{align}
The qubit is then measured in the eigenbasis of $X(\theta_{(2,2n+12)})$ where $\theta_{(2,2n+12)} = (-1)^{s_{(2,2n+12)}+1}\pi/2$.  The corresponding measurement outcome $m_{(2,2n+12)}\in\mathbb{F}_2$ is returned to the classical side-processor.
\item[(11)] For the qubits with label $(q_{\mathrm{row}},q_{\mathrm{col}}) = (j,2n+13)$ for $j\in[2]$ the classical side-processor returns the bit
\begin{align}
s_{(j,2n+13)} = \left(\bigoplus_{l=1}^{n+6}m_{(j,2l)}\right)\oplus m_{(3-j,5)}\oplus m_{(3-j,2n+11)}.
\end{align}
Each qubit is then measured in the eigenbasis of $X(\theta_{(j,2n+13)})$ where $\theta_{(j,2n+13)} = (-1)^{s_{(j,2n+13)}}\pi/2$.  Each measurement outcome $m_{(q_{\mathrm{row}},q_{\mathrm{col}})}\in\mathbb{F}_2$ is then returned to the classical side-processor.
\item[(12)] For the qubits with label $(q_{\mathrm{row}},q_{\mathrm{col}}) = (j,2n+14)$ for $j\in[2]$ the classical side-processor returns the bit
\begin{align}
s_{(j,2n+14)} = \bigoplus_{l=1}^{n+7}m_{(j,2l-1)}.
\end{align}
Each qubit is then measured in the eigenbasis of $X(\theta_{(j,2n+13)})$ where $\theta_{(1,2n+13)} = (-1)^{s_{(1,2n+13)}+1}2\theta$ and $\theta_{2,2n+13} = (-1)^{s_{(2,2n+13)}+1}\pi/4$.  Each measurement outcome $m_{(q_{\mathrm{row}},q_{\mathrm{col}})}\in\mathbb{F}_2$ is then returned to the classical side-processor.
\item[(13)] For the qubits with label $(q_{\mathrm{row}},q_{\mathrm{col}}) = (j,2n+15)$ for $j\in[2]$ the classical side-processor returns the bit
\begin{align}
s_{(j,2n+15)} = \left(\bigoplus_{l=1}^{n+7}m_{(j,2l)}\right)\oplus m_{(3-j,5)}\oplus m_{(3-j,2n+11)}.
\end{align}
Each qubit is then measured in the eigenbasis of $X(\theta_{(j,2n+15)})$ where $\theta_{(j,2n+15)} = (-1)^{s_{(j,2n+15)}+1}\pi/2$.  Each measurement outcome $m_{(q_{\mathrm{row}},q_{\mathrm{col}})}\in\mathbb{F}_2$ is then returned to the classical side-processor.
\item[(14)] For qubits $(1,2n+16)$, $(1,2n+18)$, and $(1,2n+20)$ the classical side-processor computes and returns the bits
\begin{subequations}
\begin{align}
s_{(1,2n+16)} &= \bigoplus_{l=1}^{n+8} m_{(1,2l-1)}\\
s_{(1,2n+18)} &= \left(\bigoplus_{l=1}^{n+8} m_{(2,2l-1)}\right)\oplus m_{(1,2n+17)}\\
s_{(1,2n+20)} &= \left(\bigoplus_{l=1}^{n+8} m_{(1,2l-1)}\oplus m_{(2,2l-1)}\right)\oplus m_{(1,2n+17)} \oplus m_{(1,2n+19)}.
\end{align}
\end{subequations}
Each qubit is then measured in the eigenbasis of $X(\theta_{(1,\mu)})$ where
$\theta_{(1,2n+16)} = (s_{1,2n+16}+1)\pi/2$,  $\theta_{(1,2n+18)} = (-1)^{s_{(1,2n+18)}+1} \pi/4$, and $\theta_{(1,2n+20)} = (-1)^{s_{(1,2n+20)}+1}\pi/4$.  Each measurement outcome $m_{(q_{\mathrm{row}},q_{\mathrm{col}})}\in\mathbb{F}_2$ is then returned to the classical side-processor, which then computes and returns the computational output
\begin{align}
y = m_{(1,2n+16)} \oplus m_{(1,2n+18)} \oplus m_{(1,2n+20)}.
\end{align}
\end{itemize}
It follows that $y=\mathsf{Mod}_{3,0}(\mathbf{x})$.
\end{proof}

\begin{figure}
\label{fig:mod_3_2d_mbqc}
\includegraphics[width=\linewidth]{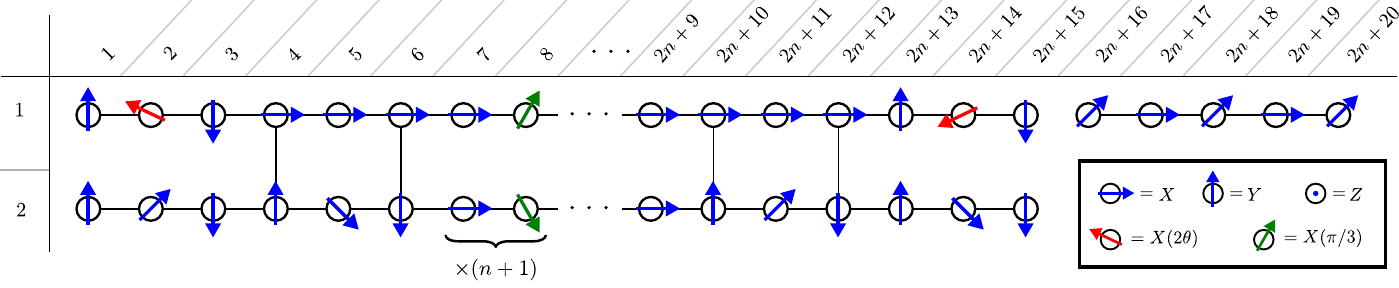}
\caption{The graph state and measurement pattern used to get a minimal implementation of the circuit in Ref.~\cite{moore1999quantum} that computes mod-3 via adaptive $l2$-MBQC.  Each qubit in the rectangular grid is given a row and column index.  Each arrow indicated the basis in which the qubit is measured, $X(\pm\theta)$ where the sign depends on an intermediate output from the mod-2 linear classical side processor.  All slanted blue arrows are $X(\pi/4)$ measurements.  }
\end{figure}

\end{widetext}

\bibliography{ref.bib}

\end{document}